\documentclass[a4paper,twocolumn,11pt,unpublished]{quantumarticle}
\pdfoutput=1
\usepackage[utf8]{inputenc}
\usepackage[english]{babel}
\usepackage[T1]{fontenc}
\usepackage{amsmath}
\usepackage{hyperref}
\usepackage{amssymb}
\usepackage{amsthm}
\usepackage{mathtools}
\usepackage{bm} 
\usepackage{stmaryrd}
\usepackage{float}
\usepackage{algorithm}
\usepackage{algpseudocode}
\usepackage{xurl}
\usepackage{hyperref} 

\usepackage{tikz}
\usepackage{lipsum}

\newtheorem{theorem}{Theorem}[section]
\newtheorem{lemma}[theorem]{Lemma}

\theoremstyle{definition}
\newtheorem{definition}[theorem]{Definition}

\newcommand{\defeq}{\overset{\text{def}}{=}}

\newcommand{\ket}[1]{\lvert #1 \rangle}

\begin{document}

\title{General circuit mapping algorithm for neutral atom quantum computers}

\author{Neven Gentil}
\affiliation{NanoComputing Research Lab, Department of Electrical Engineering, Eindhoven University of Technology, The Netherlands}
\email{n.s.gentil@tue.nl}
\author{Rianne S. Lous}
\affiliation{Coherence and Quantum Technology Group, Department of Applied Physics and science Education and Center for Quantum Materials and Technology Eindhoven (QT/e), Eindhoven University of Technology, The Netherlands}
\author{Aida Todri-Sanial}
\affiliation{NanoComputing Research Lab, Department of Electrical Engineering, Eindhoven University of Technology, The Netherlands}
\maketitle

\begin{abstract}
    Neutral atom quantum computers (NAQC) are emerging as a promising, scalable quantum computing platform because of their long qubit coherence, flexible qubit arrangement, and multi-qubit gate capabilities. However, circuit execution often requires physically moving qubits, making compilation a critical optimization challenge. We propose a circuit-independent mathematical framework built on graph-theoretic combinatorial optimization that determines the minimal number of required qubit transfers. This model captures spatial constraints specific to NAQC platforms with zone-limited gate operations and multi-qubit gates. From this framework, we encode the qubit mapping problem as a nonlinear integer program and solve it using a genetic algorithm, enabling trade-offs between minimizing the total traveled distance and the number of parallel transfer operations. Compared to the state-of-the-art scalable compiler for zoned architectures, our approach consistently finds fewer transfers. Depending on the optimization focus, our method produces shorter traveled distances or fewer parallel transfer operations. This work provides both theoretical guaranties and a practical tool for efficient, architecture-aware quantum circuit compilation. As a result, practitioners can generate hardware-aware mappings that reduce movement-induced errors and better exploit atom transfer parallelism, directly improving execution efficiency on NAQC devices.
\end{abstract}

\section{Introduction}
\label{section-intro}

\begin{figure*}
    \centering
    \includegraphics[width=0.92\textwidth]{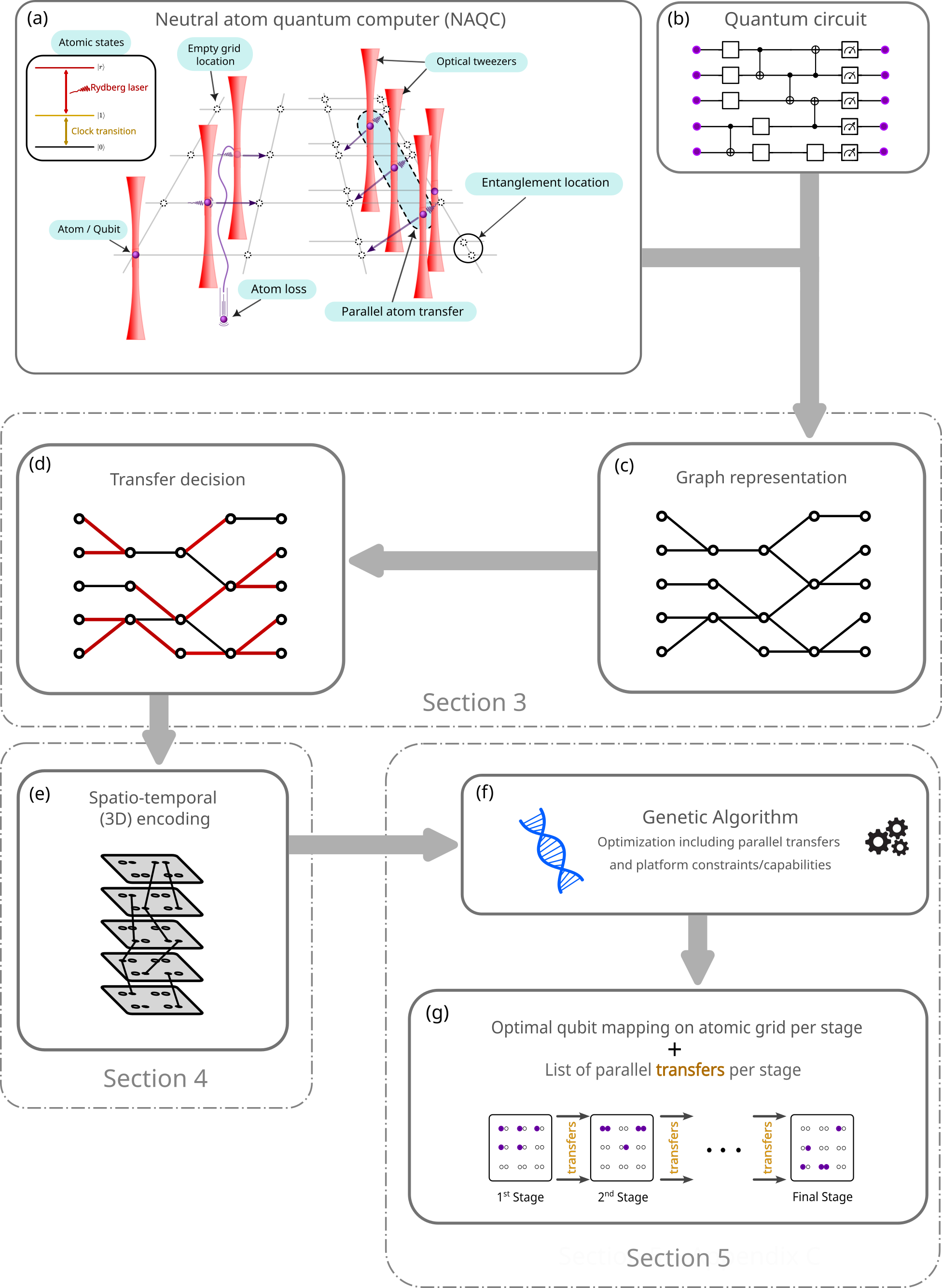}
    \caption{Schematic overview of the proposed mapper algorithm. Given a specific neutral atom platform (a) and a quantum circuit (b), the quantum digraph represents the decomposition of gate operations into different stages and atom transfers between the stages (c). This graph encompasses the specific platform constraints and capabilities. The graph representation allows to decouple qubit placement at each stage from the transfer scheme used between stages. An optimal transfer scheme is deduced (d) and used to build an efficient encoding of the mapping problem (e). Then, a genetic algorithm (f) converges to a global optimum with respect to the distance traveled by the atoms and the parallel transfers. The result is a mapping of the atoms at each stage and a sequence of parallel transfers applied between the stages (g).}
    \label{fig:graph-abstract}
\end{figure*}

NAQC platforms are emerging as a highly scalable platform at the forefront of quantum computation and quantum simulation \cite{Schmid_2024}\cite{Henriet2020}, due to their excellent connectivity \cite{Bluvstein2024} and scalability \cite{Manetsch2025}\cite{Chiu2025} . They leverage precisely controlled arrays of ultracold atoms manipulated by laser fields. In NAQC, qubits are encoded in atomic states, and the Rydberg blockade mechanism produces an equivalent CZ gate \cite{PhysRevLett.85.2208} as the basic entanglement gate. Furthermore, one advantage of Rydberg physics is the in-built availability of multi-qubit gates. The CCZ gate has already been implemented \cite{Pelegri2022}\cite{Tang2022} and, more generally, $C^k$-phase gates are very promising for NAQC due to the natural scalability of Rydberg-mediated multi-qubit interactions\cite{Pelegri2022}. In order to allow \textit{dynamic} connectivity across the atomic grid and apply entanglement gates for arbitrary pairs of qubits, the atoms are physically moved through optical tweezers. This sequence of transfers is called \textit{remapping}. Therefore, for a given quantum circuit, one needs to define where the qubits are placed and at which timestep of the circuit. This is commonly referred to as the \textit{mapping} process, and specific algorithms or software dedicated to that task are called \textit{mappers}.

Within the NISQ era, reducing the number of atom transfers is an important challenge for the NAQC platforms, which improves the connectivity, yet requires dedicated software and algorithms. Early approaches \cite{Tan_2024} \cite{Ludmir2024Parallax} \cite{Wang2025Atomique} focused on optimizing and mapping the location of atoms throughout the quantum circuit using satisfactory modulo theory or rudimentary greedy algorithms. They exploited the atom movements to overcome scalability issues with swap gates. They focused on monolithic architectures (M-arch) where gates are executed in the same region of space. However, for realistic quantum circuits up to a hundred of qubits, the total fidelity induced by the mapping process was often too low for actual testing in experiments. In 2024, Enola mapper \cite{10.1145/3658617.3697778} introduced simulated annealing for qubit placement and graph optimization techniques for qubit routing during transfers. The total fidelity improved by more than five times and quantum circuits of up to 10,000 qubits were practically mapped within 30 minutes. 

Later, PowerMove \cite{ruan2025powermove}, ZAP \cite{huang2024zapzonedarchitectureparallelizable} and Mantra \cite{jang2025qubitmovement} mappers introduced the utilization of zoned architectures (Z-arch). In this architecture \cite{Bluvstein2024}, the entanglement zone is physically separated from the storage zone preventing idling qubits from being excited to a Rydberg state. Then ZAC \cite{lin2025reuseawarecompilationzonedquantum}, a reuse-aware Z-arch mapper, incorporated graph optimization techniques for qubit placement prior to a simulation annealing (SA) process.  In particular, ZAC relies on bipartite graphs to maximize the reusability of the qubits between the stages and the different zones of the quantum computer. Coupled to the SA algorithm, this stochastic global optimization process samples the solution space via a temperature-controlled acceptance probability, allowing to converge toward global optimum for qubit mapping. Finally, a mapper from the \textit{Munich Quantum Toolkit} (MQT) \cite{Stade2025ICCAD} introduced a routing-aware placement method, jointly optimizing qubit mapping across stages and qubit parallel movements. This mapper succeeds his predecessor NALAC \cite{Stade2024QCE} and compares its benchmark with ZAC, claiming fewer rearrangement steps per circuit.  Although these approaches improve placement and reduce movement costs, they did not provide explicit lower bounds or a general formulation for optimal atom transfer.

In this work, we go beyond existing heuristic NAQC mappers by introducing a unified graph-theoretic framework that evaluates the minimum number of physical transfers required to run a quantum circuit on both M-arch and Z-arch with any quantum gate set \cite{Gentil2026Patent}, including multi-qubit gates as the CCZ gate. After providing an overview of the proposed mapper algorithm tailored to the capabilities of the NAQC platform in Section \ref{section-naqc}, we introduce graph theory to map any quantum circuit to a graph representation in Section \ref{section_transfers}. Building on our theoretical results, we introduce in Section \ref{qubit-mapping-encoding} an efficient encoding of the mapping problem and a genetic algorithm that jointly optimizes the total distance traveled by the atoms and the use of parallel transfers. In Section \ref{section-transfer-operation} we establish a systematic approach using graph representation to efficiently construct PTOs as a post-process of a given mapping. This post-process construction coupled with the efficient encoding reveals a controllable trade-off between parallelism and spatial efficiency. In Section \ref{section-results} we compare our approach with ZAC and MQT showing that the algorithm performs on par with or exceeds the state-of-the art. Our work provides both a mathematical baseline for NAQC compilation and a flexible optimization tool adaptable to various architectures.

\section{Neutral Atom Quantum Computer (NAQC) And Algorithm Overview}
\label{section-naqc}

A schematic overview of the full algorithm is sketched in Figure \ref{fig:graph-abstract}, highlighting the NAQC capabilities in Figure \ref{fig:graph-abstract}-(a). Typically, a NAQC platform is a 2D grid of individually trapped atoms that can be programmed by laser pulses and/or microwaves \cite{Wintersperger2023} \cite{Henriet2020} \cite{Morgado2021}. The atoms are prepared in specific quantum states within an ultra-high vacuum apparatus and kept in place  with the aid of optical tweezers or lattices. Commonly used atom species are Cesium, Rubidium, or Strontium, and can come in single species setups or in heterogeneous mixtures \cite{Saffman2010}\cite{Madjarov2020NatPhys}\cite{Anand2024DualSpecies}\cite{venderbosch2026robuststrontiumtweezerapparatus}. The platform also allows 3D grid configurations \cite{Barredo2018}. 

Neutral atom qubits use two atomic levels as $\ket{0}$ and $\ket{1}$, driven by resonant optical or microwave fields inducing Rabi oscillations. Single-qubit gates are implemented via pulses with controlled duration, phase, and frequency, and readout is performed through state-selective fluorescence detection. Coherence times are on the order of several seconds. Typically, the coherence time of the state $\ket{1}$ for rubidium platforms is about $T_2\simeq1.5\text{s}$ \cite{Evered2023}. The entanglement gates in NAQC use a Rydberg state $\ket{r}$ and a second Rydberg laser, which is typically applied globally over hundreds of micrometers. However, the Rydberg blockade mechanism requires that the two atoms involved in a CZ gate \cite{PhysRevLett.123.170503} \cite{PhysRevLett.85.2208} be close enough to maximize the atom-atom interaction. The implementation of the gate therefore relies on physically bringing the selected qubits into proximity via atom transfer. The distance required to entangle two Rydberg atoms is on the order of 2\,µm to 4\,µm for current platforms\cite{Wintersperger2023}. The Rydberg interaction, which is an van der Waals type interaction, scales as $R^6$ over the distance $R$ between atoms \cite{PhysRevLett.110.263201} and is negligible for a distance greater than 10 \,µm. Otherwise, the qubits separated by a distance shorter than this range are susceptible to undergo crosstalks that lead to infidelity of the CZ gate. 

When executing an quantum circuit \ref{fig:graph-abstract}-(b), the entanglement gates require movement of qubits and thereby limit the fidelity of quantum circuit execution. Typically, the probability of successfully transferring an atom to a target position is greater than 0.99 for recent platforms\cite{PhysRevA.106.022611}\cite{Barredo2016}\cite{Schymik2022}. Therefore, after hundreds of  moves, the probability of losing an atom or changing their state during the movement is greater than one chance over two. This limit is often rapidly reached for quantum circuits involving as few as 100 qubits.

Transferring atoms in space requires time, which reduces the computation time of the quantum circuit. It is typically done using moveable tweezers controlled through acousto optical deflectors. The speed of the tweezers varies according to the platform but can be as high as $55 \text{µm/µs}$\cite{Bluvstein2022}. With this technology, the rows and columns of atoms can be moved simultaneously, allowing parallel transfer operations (PTOs). However, constraints apply and atoms can not cross already filled sites.

To create an optimal sequence of parallel transfers, we use a graph representation (Figure\ref{fig:graph-abstract}-(c)) that decomposes the quantum circuit into stages with parallel transfer operations in between. This allows us to select which atoms need to be rearranged at specific times, leading to a transfer decision graph (Figure\ref{fig:graph-abstract}-(d)). The degree of freedom left by the transfer decision permits one to construct an efficient spatio-temporal encoding of the quantum circuit mapping problem, Figure \ref{fig:graph-abstract}-(e). Then, a genetic algorithm optimizes both the global distance traveled by the atoms and the amount of PTOs, as shown in Figures \ref{fig:graph-abstract}-(f). Thus, an optimized mapping is obtained with a list of parallel transfers per transition (Figure \ref{fig:graph-abstract}-(g)).

\section{Graph Representation}
\label{section_transfers}

\begin{figure}
    \centering
    \includegraphics[width=0.92\linewidth]{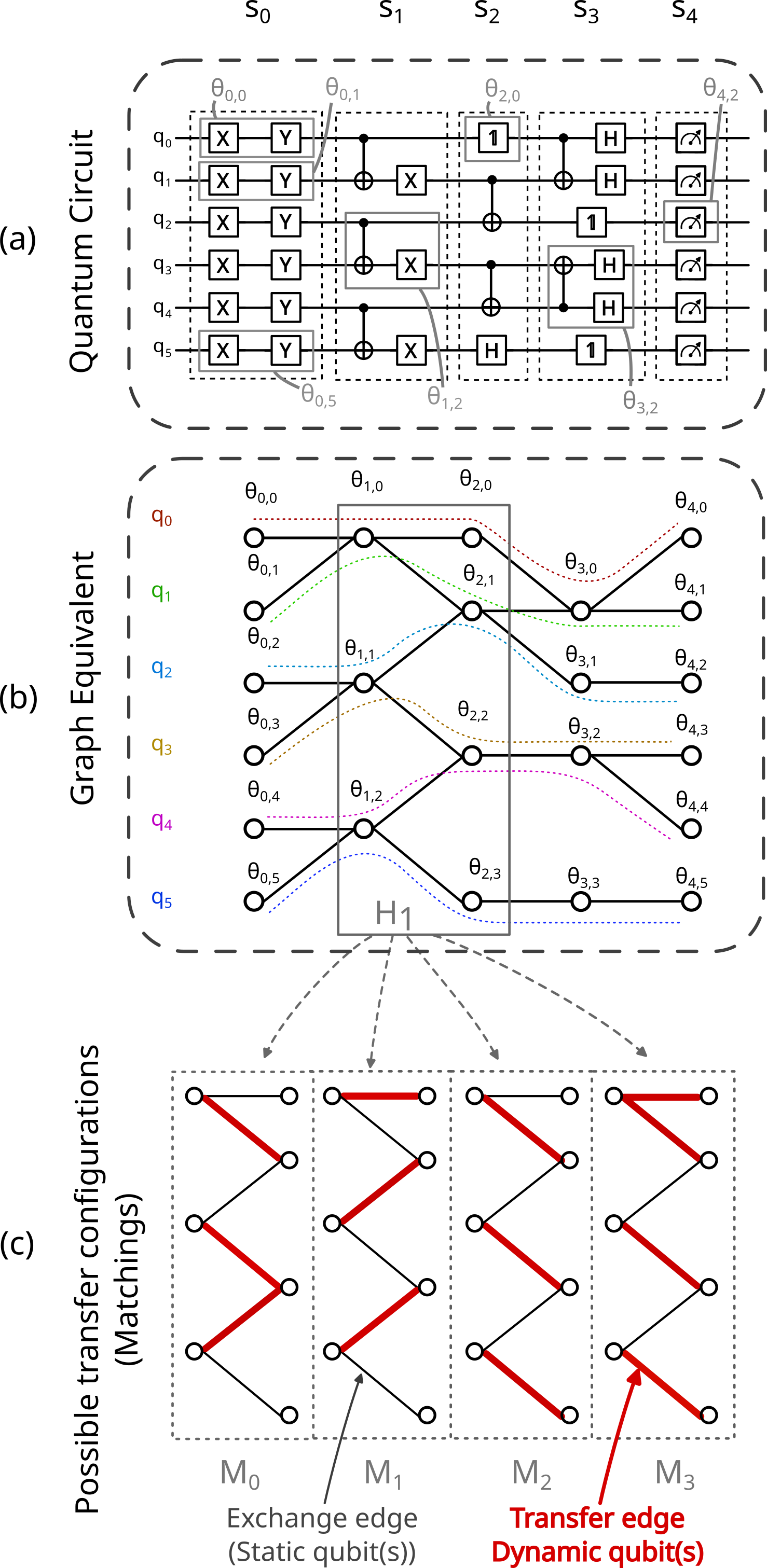}
    \caption{Example of quantum circuit transformed into a quantum circuit digraph with extraction of possible matchings between consecutive stages. (a) is an example of scheduled quantum circuit. The stages are labeled $s_i$. Inside the stages, the series of gates affecting common qubits are grouped into instructions $\theta_{i, j}$. Transfers of the qubits, or \textit{remapping}, are executed between each stage exclusively. (b) is the associated general quantum digraph $G$ from (a). The edges are directed from $\theta_{i,j}$ to $\theta_{i + 1, j'}$. The instructions applied on the qubits along the stages are represented by colored dotted lines: $q_0$ undergoes $\theta_{0, 0}$, $\theta_{1, 0}$, $\theta_{2, 0}$, $\theta_{3, 0}$ and $\theta_{4, 0}$ for example. The subgraph $H_1$ is \textit{extracted} from successive stages $s_1$ and $s_2$ in $G$. In particular $H_1$ is the subgraph containing the vertices $\theta_{1, j}$ and $\theta_{2, j'}$. (c) represents four possible matchings for the subgraph $H_1$. The edges marked as transfers are in red. The edges marked as static are in black. The matching is an edge coloring (black color) of the bipartite graph $H_1$. Note that $M_0$, $M_1$ and $M_2$ are maximum weighted matchings. They can be labeled $\mathcal{M}_{1,0}$, $\mathcal{M}_{1,1}$ and $\mathcal{M}_{1,2}$ in this example to indicate the connection with $H_1$.}
    \label{fig:qcirctom}
\end{figure}

As a first step, the algorithm decomposes a given quantum circuit into a graph representation, as illustrated in Figure \ref{fig:qcirctom}-(a). The quantum circuit consists of a register of a fixed number of qubits initialized in given quantum states, a time-ordered series of quantum gates, and a measurement gate applied to each qubit at the end of the circuit. The number of qubits is $N_q$ and a qubit is denoted $q_l$ where $l \in \llbracket0; N_q-1\rrbracket$. Given a quantum circuit, we assume that the associated time-ordered series of quantum gates has been optimized (i.e. twice CNOT is the identity) and written with the quantum gates available in the library of the targeted platform. Thus, we assume that a quantum circuit has a unique representation.

We divide the quantum circuit into \textit{stages} and \textit{instructions}. A stage is defined as \textit{a series of quantum gates in which the register does not need to be remapped}. The number of stages is $N_{s}$ and a specific stage is designated by $s_i$ where $i \in \llbracket 0; N_{s} - 1 \rrbracket$. We call an \textit{instruction} the series of quantum gates applied to a group of qubits for a given stage. The number of instructions involved inside a stage $s_i$ is $N_{\theta,i}$ and a specific instruction is written as $\theta_{i, j}$ where $j \in \llbracket 0; N_{\theta,i} - 1\rrbracket$. Each instruction $\theta_{i, j}$ has a set of qubits associated to it. For a given stage, every idle qubit is assigned to an artificial identity instruction. This instruction has no physical impact on the target, and his purpose is to analytically ensure the isolation of the qubit from other instructions. This means that for every stage, each qubit is associated with exactly one instruction. Note that we consider the quantum circuits to be optimized so that the highest possible number of quantum gates is executed per stage. We call this a scheduled quantum circuit. Any arbitrary circuit can be mapped onto this using additional scheduling software.

We define $G(V, E)$ as the general quantum digraph representing the circuit, as visualized in Figure \ref{fig:qcirctom}-(b). The edges connect different instructions between stages, for example, $\theta_{i,j}$ to $\theta_{i,j'}$. The edges are called exchange edges when the qubits concerned are static and do not physically move between the connected stages. When the qubits within undergo physical movement, we call it a transfer edge. Moving qubits are called dynamic qubits. The set of vertices in $V$ is the set of instructions $\theta_{i, j}$.  The digraph $G$ is similar to a Directed Acyclic Graph (DAG) representation, except that the vertices in $G$ contain full instructions instead of single gates for a standard DAG.

From the construction of G, we identify the subgraphs $H_i$, for which we determine possible transfer configurations called matchings ($M_i$) as shown in Figure \ref{fig:qcirctom}-(c). The subgraphs $H_i$ represent the instructions connected from one stage to another, weighted by the amount of qubits $\omega(e)$ for each edge. A matching $M_i$ is a set of edges of $H_i$ such that no two edges share a common vertex. Those edges are either an exchange or a transfer edge, as depicted with red or black coloring. The challenge in defining $M_i$ is to find a physically valid and optimal set of transfer edges having the least number of underlying dynamic qubits. We prove in Appendix \ref{gt-of-opt-transf} that a transfer configuration between $s_i$ and $s_{i+1}$ corresponds to a matching of $H_i$, noted $M_i$. Note that it is physically impossible to associate two qubits of different instructions in $s_i$ with static edges, while both qubits need to reach the same instruction, and hence the same location in $s_{i+1}$.

The score of any matching $M_i$ is defined as the summation of the amount of static qubit assigned by the edges in $M_i$: 
\begin{equation}
\Omega(M_i)= \sum_{e \in M_i} \omega(e).
\end{equation} 
In the case where the upper limit $\Omega(M_i)$ is reached for a certain $H_i$, this maximum matching $M_i$ is renamed $\mathcal{M}_i$. The maximum weighted matching $\mathcal{M}_i$ is not necessarily unique, and we label them $\mathcal{M}_{i,0}$, $\mathcal{M}_{i,1}$, etc. Physically, it means that given a couple $(s_i, s_{i+1})$, there are multiple valid transfer configurations that minimize the number of transfers.

We arbitrarily pick $\mathcal{M}_i$ in $\mathcal{M}_{i,0}$, $\mathcal{M}_{i,1}$, etc. because only the score of the chosen matching matters. In other words, the global optimum of the transfers is obtained by concatenating the local optimum found for each transition $H_i$. For M-arch, this leads to the following theorem:
\begin{theorem}
\label{theorem-global-opt-m}
    The total number of transfers $\lambda_q$ for a quantum circuit $G$ for M-arch is minimized by the sum of each of the local minimums in $\mathcal{M}_i$:
\begin{equation}
    \lambda_q \defeq \sum_{i=0}^{|H(i)|} \Omega(\mathcal{M}_i)
\end{equation}
\end{theorem}
A proof of Theorem \ref{theorem-global-opt-m} can be found in Appendix \ref{gt-of-opt-transf}. Numerically, the Hungarian algorithm with inverted weights can be used to find the maximum weighted matchings $\mathcal{M}_{i,j}$ \cite{KorteVygen2018}.

\begin{figure}
    \centering
    \includegraphics[width=0.68\linewidth]{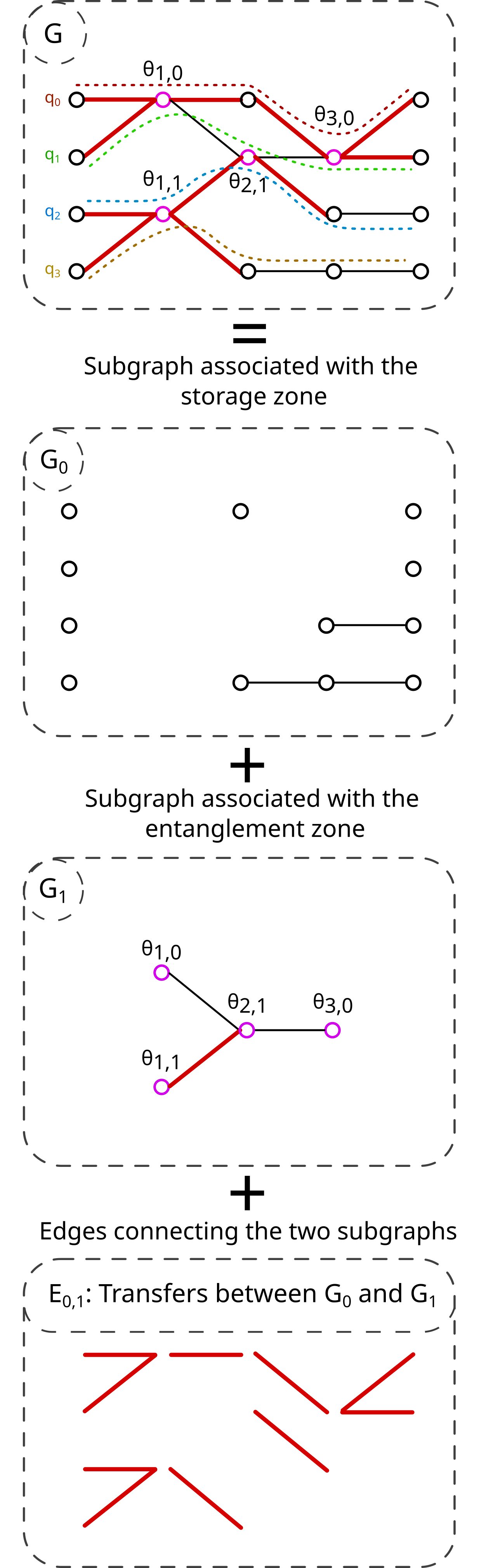}
    \caption{Decomposition of part of the quantum circuit from Figure \ref{fig:qcirctom} for a zoned architecture with two zones. The edges marked as transfers are in red. The edges marked as static are in black. The instructions containing two qubits are pink circles. In this example, $\hat{\theta}_{Z_1} = \{\theta_{1,0}, \theta_{1,1}, \theta_{2,1}, \theta_{3,0}\}$. $G$ is therefore the union of the subgraph $G_0$ containing single-qubit gates only, $G_1$ containing two-qubit gates/single qubit gates, and the edges $E_{0,1}$ connecting $G_0$ and $G_1$. Since the physical zones associated with $G_0$ and $G_1$ are spatially separated, the edges connecting $G_0$ and $G_1$ are necessarily transfers.}
    \label{fig:zoned-arch-ex}
\end{figure}

Theorem \ref{theorem-global-opt-m} can be generalized for Z-arch. We consider $N_Z \geqslant 1$ zones for specific quantum hardware. The zones are labeled $Z_0$, $Z_1$, etc. Given that a quantum digraph G is executable on this hardware, every instruction $\theta_{i,j}$ is strictly associated with exactly one zone. Thus, the instructions are grouped into sets $\hat{\theta}_{Z_0}$, $\hat{\theta}_{Z_1}$, etc. where $\hat{\theta}_{Z_0}$ owns all the instructions strictly executed in zone $Z_0$, for instance. In the following, $m$ and $n$ are variables in $\llbracket 0; N_Z - 1\rrbracket$.
\begin{theorem}
\label{theorem-glob-opt-subgraph}
    Considering a quantum circuit executed on a Z-arch, the minimum number of transfers in the global quantum digraph $G(V,E)$ is obtained by minimizing the transfers in all subgraphs $G_n(V_n, E_n)$ with $V_n \defeq \hat{\theta}_{Z_n}$ and $E_n \defeq \{\{u, v\} \in E | u \in V_n, v \in V_n\}$.
\end{theorem}
As an example, in Figure \ref{fig:zoned-arch-ex} we consider a subcircuit of Figure \ref{fig:qcirctom}-(a,b) executed on a $N_z=2$ Z-arch (2Z-arch), i.e. single-qubit instructions are localized in a zone different from the two qubit instructions. Note that the minimum amount of transfers $\lambda_q$ for Z-arch is not strictly the direct sum of the minimum amount of transfers for each $G_n$. A proof of Theorem \ref{theorem-glob-opt-subgraph} and the actual expression for the minimum amount of transfers can be found in Appendix \ref{gt-of-opt-transf}. 

In Figure \ref{fig:zoned-arch-ex}, we see that the maximum weighted matchings are found in $G_0$ and $G_1$ in order to find the global optimum. The set of edges in $E_{0,1}$ are \textit{mandatory} transfers: atoms are moved from/to different zones of the hardware, and the edges involved cannot be selected as \textit{exchange}.

\section{Qubit Mapping Encoding}
\label{qubit-mapping-encoding}

\begin{figure}
    \centering
    \includegraphics[width=0.90\linewidth]{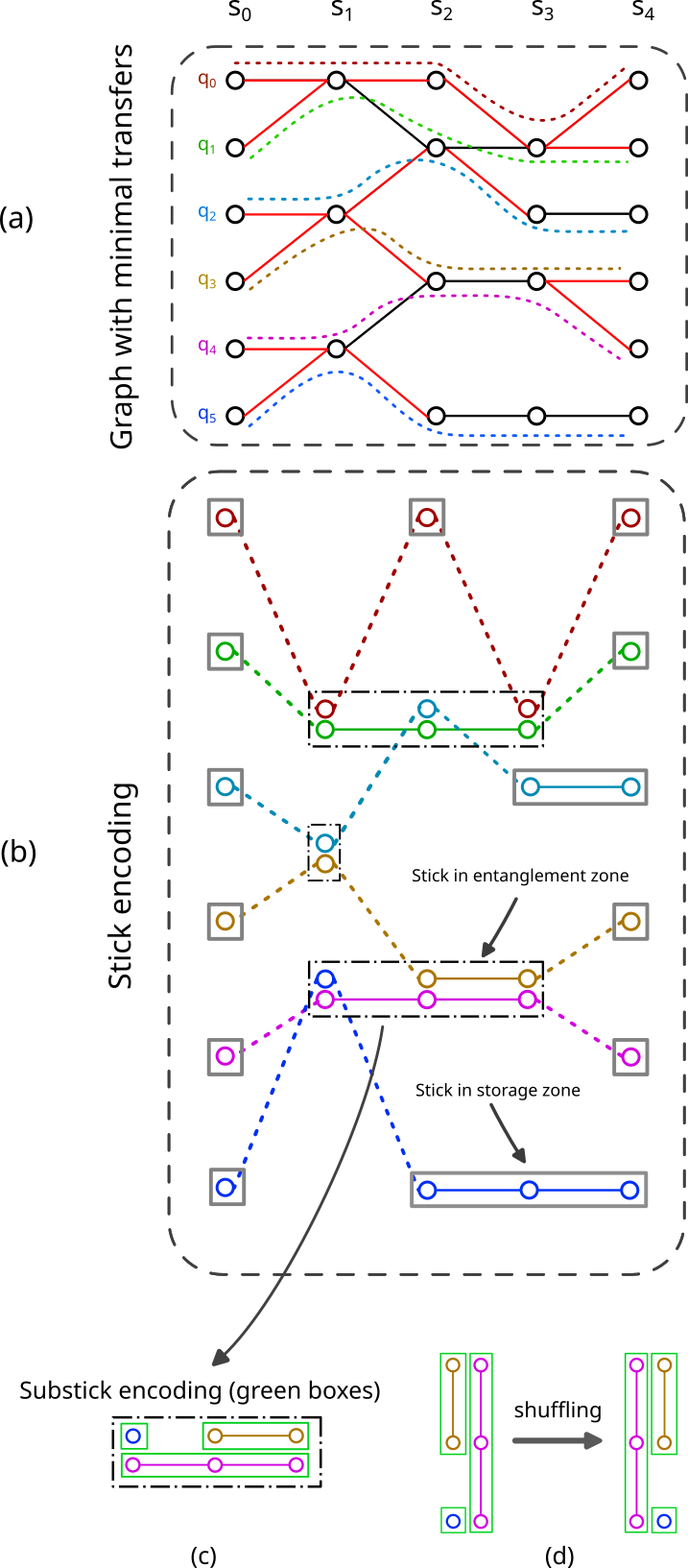}
    \caption{Transforming the graph representation into stick encoding. Sub-figure (a) is the digraph $G$ of the quantum circuit in Figure \ref{fig:qcirctom} where the matchings (black edges) fit the minimum number of transfers for a Z-arch. The red edges are transfers: 16 transfers are required to run the circuit on a two-zoned architecture. Sub-figure (b) is the efficient stick encoding of $G$. The black boxes are sticks associated with $G_0$ while dashed black boxes are sticks associated with $G_1$ (see Figure \ref{fig:zoned-arch-ex}). The dotted lines are transfers. The colors correspond to the qubit colors in (a). Sub-figure (c) is an example of substicks encapsulation (green boxes). Sub-figure (d) is an example of shuffling of the substicks (green boxes) for a single stick which allows for modifying the stick encoding. In this specific case, only two-qubit gates are involved. Hence, as the only shuffling operation available, the two substicks on the left (gold and blue nodes) can have their position on the atomic grid swapped with the substick on the right (pink nodes).}
    \label{fig:g2sticks}
\end{figure}

\begin{figure*}
    \centering
    \includegraphics[width=1.0\textwidth]{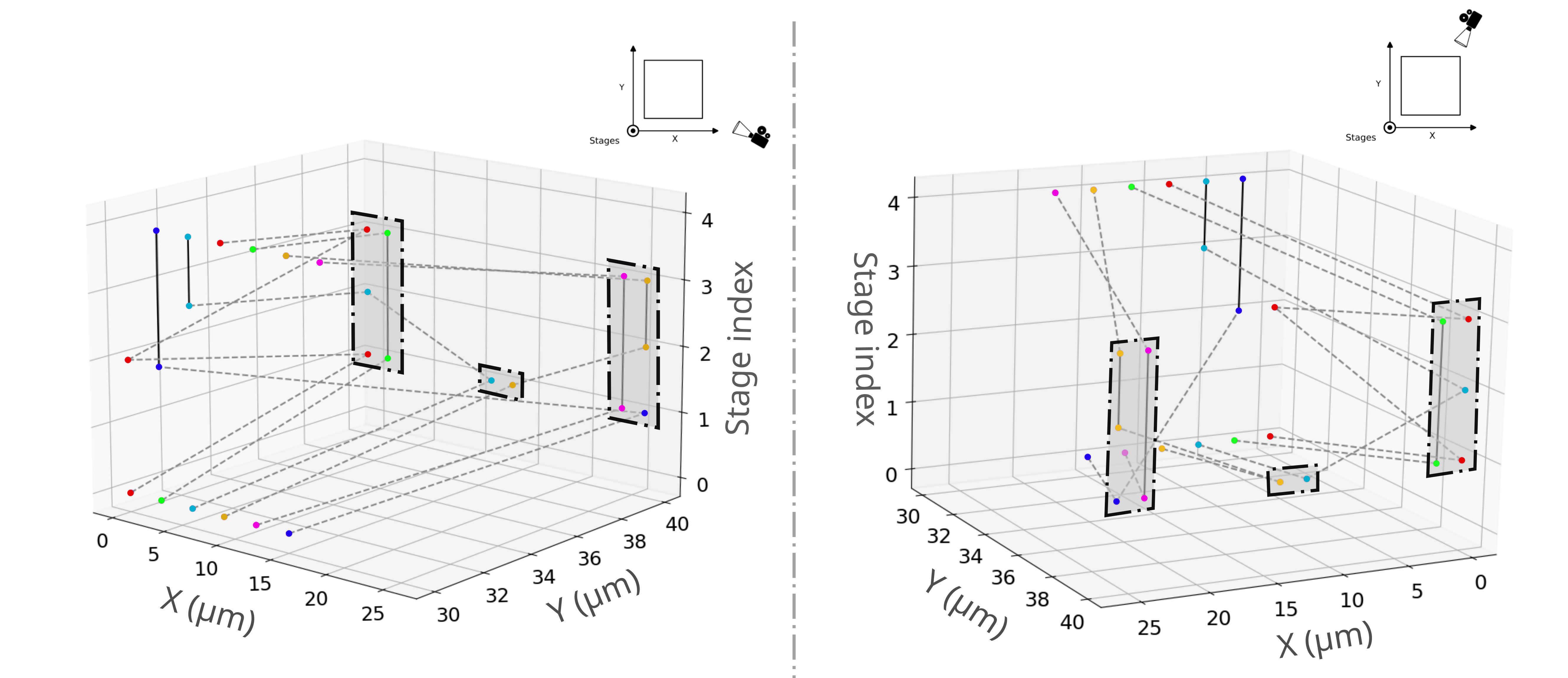}
    \caption{A 3D view of the stick encoding described in Figure \ref{fig:g2sticks} for a 2Z-arch. The atomic grid is the XY plane. The Z axis is time / stages. The atoms are colored according to Figure \ref{fig:qcirctom}. The transfers are dotted grey lines connecting the atoms. The black lines represent the substicks. We clearly see the atoms in the storage zone at Y=30, Stage index=0 and Y=30, Stage index=4. The entanglement zone is located along the Y=40 vertical plane. Dashed black boxes encompass sticks/substicks involved in the entanglement zone.}
    \label{fig:encoding-3d}
\end{figure*}

From the graph representation, we extract an encoding that allows us to optimize the actual mapping of the qubits for each stage. In this section, we present an encoding, that is, a representation of the circuit, that exposes only the intrinsic degrees of freedom (DoFs) after a transfer configuration has been found (Section \ref{section_transfers}). For each \textit{stick} described in the following encoding, the associated DoFs are a 2D spatial coordinate on the atomic grid, a combination of the underlying \textit{substicks}. Such encoding will then allow to efficiently optimize the traveled distance of every atom and the possible PTOs, hence increasing the total fidelity of the circuit.

In Figure \ref{fig:g2sticks}-(a) we show the optimal transfer configuration for the quantum circuit in \ref{fig:qcirctom}-(a) to be executed on a 2Z-arch. Figure \ref{fig:g2sticks}-(b) is the associated efficient encoding of the remaining DoFs. By representing the "grid" on the X axis and the stages of the circuit on the Z axis, sets of successive edges marked as transfer get a distinguishable \textit{stick-like} shape. We call them \textit{sticks} and name this specific representation the \textit{stick encoding}. A stick with a single qubit per stage is associated with a unique position on the atomic grid: all the underlying instructions are located at the same position for the respective stages. To avoid physical collision, each qubit of a given multi-qubit instruction $\theta_{i,j}$ must have its own dedicated position on the grid at stage $s_i$. We consider a distance of 2µm to be sufficient to avoid collisions, which is typically found for current NAQC platforms \cite{Wintersperger2023}. Therefore, for sticks with two (or more) qubits per stage, a \textit{position} is actually a pair (or more) of neighboring atomic grid locations. Physically, changing the position of any stick changes the position of all the involved qubits among all the underlying stages.

To complete the encoding, it suffices to add the edges marked as transfer (dashed lines in Figure \ref{fig:g2sticks}-(b,c)). Transfers start from an instruction on a stick in a given stage $s_i$ and end with an instruction on another stick in stage $s_{i+1}$. As we can see in Figure \ref{fig:g2sticks}-(b), some qubits stay in the same stick for one or more stages. They form \textit{substicks} (green boxes in Figure \ref{fig:g2sticks}-(c)). The substicks belonging to the same stick cannot be separated from each other. 

Modifying the relative position of the substicks inside a given stick consists of an additional DoF, which can be exploited to further optimize the qubit mapping. The only way to internally modify a stick is to \textit{shuffle} its substicks, relative to the pair (or more) of neighboring grid locations. An example is given in Figure \ref{fig:g2sticks}-(d). Depending on the topology of the sticks, only a few possible combinations are available inside each stick. In particular, sticks containing two qubits per stage can only be swapped (into a given stick, all the left sided substicks are swapped with all the right sided substicks), while sticks containing a single qubit per stage cannot be shuffled at all. Thus, we can refer only to a list of combinations of a stick rather than to the full positioning of the substicks. In particular, sticks containing two qubits per stage have a single boolean variable indicating if the stick is swapped. This problem is addressed in more detail in the Appendix \ref{substick-combination}.

Figure \ref{fig:encoding-3d} shows a 3D view of the efficient encoding for a 2Z-arch introduced in Figure \ref{fig:g2sticks}-(b), where the vertical axis represents the different stages over time. The figure highlights the storage and entanglement zones, as well as the atom transfers between stages It shows that the degree of freedom for each stick is composed of a discrete 2D position plus a Boolean variable, i.e. the \textit{swap} state, for sticks with two qubits per stage. Note that both components are based on integers. In our algorithm, particular attention is given to the collision avoidance of the sticks, i.e. the discrete 2D position is not fully free.

We use a genetic algorithm (GA) without crossover to minimize the total traveled distance of the atoms and maximize the number of PTOs among the transitions. At each generation, every solution in the population is independently randomized to produce several candidate children, and only the highest-scoring solutions are retained for the next iteration. Since no crossover operation is used, exploration of the search space relies exclusively on local random modifications of the stick configurations, such as moving or internally shuffling sticks. More details are given in Appendix \ref{ga-details}. 

\section{Parallel Transfer Operations}
\label{section-transfer-operation}

\begin{figure}
    \centering
    \includegraphics[width=0.8\linewidth]{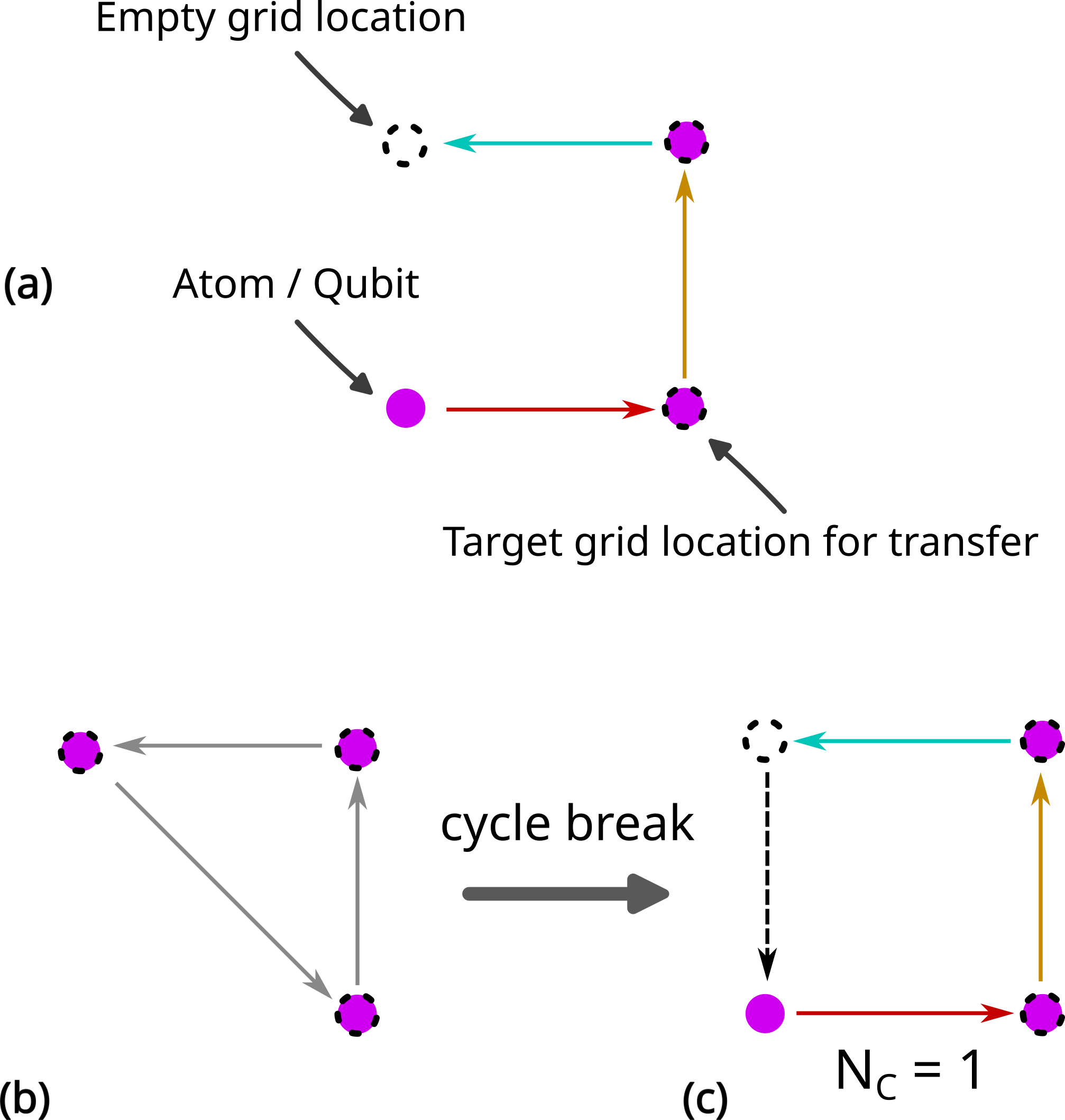}
    \caption{Simple $G_T$ without cycle (a), with cycle (b), and once the cycle is broken (c). The atoms are purple dots. The ending positions are dashed black circles. The transfers are represented by directed edges / arrows. In sub-figure (a), the time-ordered sequence of transfers is: \textit{cyan edge}, \textit{gold edge} and \textit{red edge}. Sub-figure (b) is a simple example of $G_T$ where the unique path is a cycle. The edges are gray: no transfer can be executed without physical collision of the atoms. Sub-figure (c) is the breaking process of the cycle in order to make the transition executable on the hardware. The dotted black edge is the additional transfer which must be processed ahead of any other transfer. Here the number of additional transfers is $N_C=1$, since there is only one transfer graph $G_T$. Subsequently the same sequence as in sub-figure (a) can be executed.}
    \label{fig:cycles}
\end{figure}

Using the stick encoding we determine the parallel transfer operations, give the atom transport constraints. These constraints limit the movements to avoid collisions between atoms. For that we define the transfer digraph $G_{T, i}(V_{T, i}, E_{T, i})$ for the transition $s_i \rightarrow s_{i + 1}$, abbreviated $G_T$ to indicate an arbitrary transition. The vertices in $V_{T, i}$ are the starting positions $\overrightarrow{a_l}=(a_{l,x}, a_{l,y})$ of each qubit $l$ in stage $s_i$, as well as the ending positions $\overrightarrow{b_l}=(b_{l,x}, b_{l,y})$ of the same respective qubits in stage $s_{i + 1}$. The edges in $E_{T, i}$ are the physical transfer of those qubits.

Any transfer digraph $G_T$ follows one of the two cases shown in Figure \ref{fig:cycles}. Figure \ref{fig:cycles}-(a) shows an example of $G_T$ as a single non-intersecting path. Every purple dot is an atom that needs to be physically moved to a different grid location (dashed black circles). Figure \ref{fig:cycles}-(b) shows an example of $G_T$ as an intersecting path, here a \textit{cycle}, and Figure \ref{fig:cycles}-(c) shows the final graph once the cycle is broken. In the first case (Figure \ref{fig:cycles}-(a)), a simple recursive algorithm can build the sequence of transfer operations required for the transition. In the second case, Figure \ref{fig:cycles}-(b), the cycle represented by successive transfers must be broken by transferring one of the atoms to an empty position before any other transfer operation. Every cycle requires a single additional transfer to be broken. We define $N_C$ as the number of cycles among all transitions, which also corresponds to the amount of additional transfers for the whole quantum circuit. More details on \textit{cycle breaking} can be found in the Appendix \ref{appendix-gt-for-to}.

For routing the atoms, we assume that there is enough space to allow a \textit{parking space} in addition to the \textit{travel space}. The parking space is an additional empty atomic slot reserved for temporary transfer during a transition, next to every other atomic slot. This parking space prevents finding an empty slot in the grid for the intermediate position $\overrightarrow{x}_l$ of the additional transfers. The travel space is the space reserved on the grid to freely move the atoms over the grid, which is also often assumed in other NAQC compilers \cite{lin2025reuseawarecompilationzonedquantum}. We see in Figure \ref{fig:cycles}-(c) that this procedure avoids any physical collision in between the atoms. See Appendix \ref{appendix-grid-routing} for further details on atom routing.

Once every cycle is resolved for a given $G_T$ and a given qubit mapping, it is possible to estimate the theoretical minimum number of qubit transfers to execute the related transition. In fact, $\lambda_q$ defined in Theorem \ref{theorem-global-opt-m} gives the lower limit in the number of transfers for any mapping of the quantum circuit. The following theorem estimates the minimal amount of transfers once the qubits are mapped on the atomic grid:

\begin{theorem}
\label{theorem-break-cycles}
    Given a mapping of qubits for all stages, the minimal amount of transfers required to run a scheduled quantum circuit $G$ is $\lambda_q + N_C$, where $N_C$ is the total number of cycles for all transitions.
\end{theorem}

From the transfer digrapsh, the PTOs are determined by taking into account the constraints of the atom transport. As an example, let us only consider the constraints along the X-axis of the grid from Figure \ref{fig:encoding-3d}. For two atoms $q_l$ and $q_{l'}$ transferred from positions $(a_{l,x}, a_{l,y})$ and $(a_{l',x}, a_{l',y})$ to $(b_{l,x}, b_{l,y})$ and $(b_{l',x}, b_{l',y})$, respectively, the transfers are compatible if and only if one of the following conditions is met: 
\begin{align}
a_{l,x} > a_{l',x} &\wedge b_{l,x} > b_{l',x} \label{aod_constraint_a} \\
a_{l,x} < a_{l',x} &\wedge b_{l,x} < b_{l',x} \label{aod_constraint_b}\\
a_{l,x} == a_{l',x} &\wedge b_{l,x} == b_{l',x} \label{aod_constraint_c}
\end{align}
Then, for transfers to be compatible in the XY plane, one of those three conditions must also hold for the Y axis. The constraints of the AODs on the optical tweezers allow simultaneous transfers to be compatible under the conditions listed in the relations \ref{aod_constraint_a}, \ref{aod_constraint_b} and \ref{aod_constraint_c}. Using the constraints from the AOD and avoiding atom collisions, we build and greedy algorithm that builds the PTOs. We describe the logic sequence of action evaluated by this greedy algorithm in Appendix \ref{pto_sequence}.

\section{Results and Comparison}
\label{section-results}

\begin{figure*}
    \centering
    \includegraphics[width=1.0\textwidth]{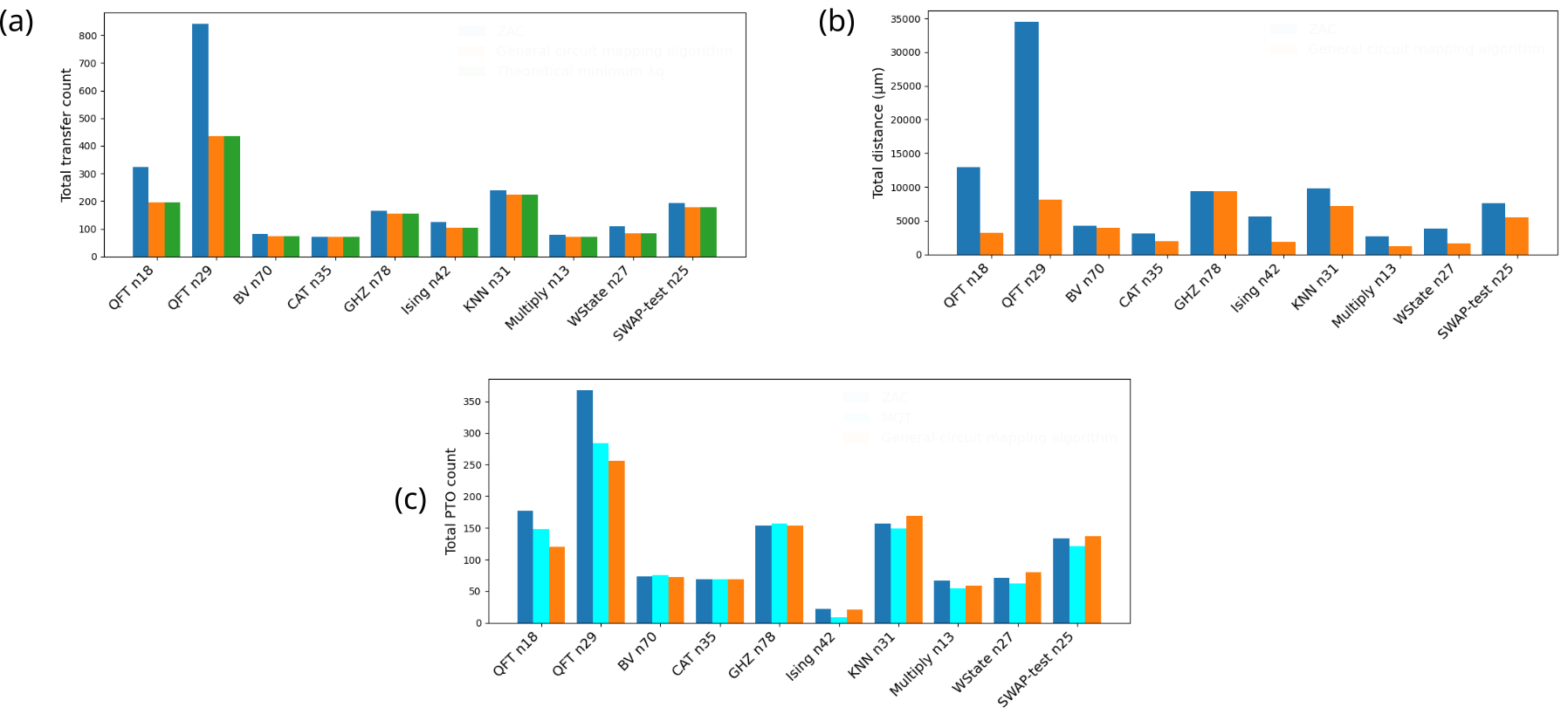}
    \caption{Results of the benchmark for ten different quantum circuits. Our mapping algorithm is in orange, ZAC is in dark blue. Sub-figure (a) shows the total transfer count for each circuit. The theoretical minimum that can be reached is in green. Sub-figure (b) shows the total distance traveled by the atoms for each circuit. Sub-figure (c) shows the total PTO count required to run each of the ten quantum circuits. MQT is in cyan.}
    \label{fig:results}
\end{figure*}

\begin{figure*}
    \centering
    \includegraphics[width=1.0\textwidth]{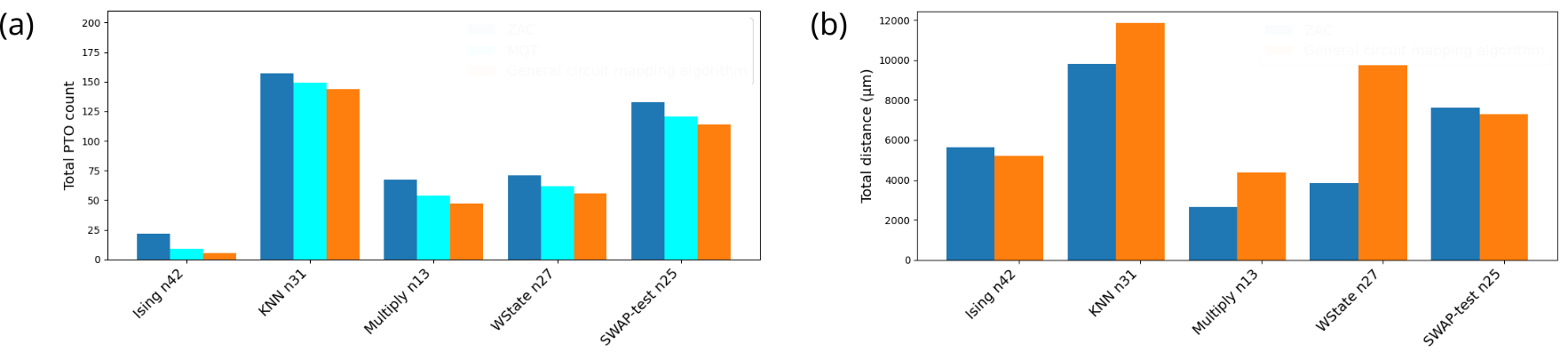}
    \caption{Benchmark results for different GA parameters. For Ising circuit, the GA is not executed. For KNN, Multiply, WState, and SWAP-test circuits, the GA measures the inverse of the PTO count as the score instead of the inverse of the total distance. Our mapping algorithm is in orange, ZAC is in dark blue. Sub-figure (a) shows the total PTO count required to run each quantum circuits. MQT in cyan. Sub-figure (b) shows the total distance traveled by the atoms for each circuit.}
    \label{fig:results-special}
\end{figure*}

The performance of the circuit mapping algorithm we developed in the previous section can be evaluated numerically for different quantum circuits and compared to the ZAC and MQT mappers. The ZAC and MQT features make these two mappers relevant for a benchmarking of the proposed general mapping solution, since, to our knowledge, they are the most advanced compiler for NAQC platforms considering zoned architectures. The configuration of the emulated hardware is a 2Z-arch composed of an entanglement and storage zones without a dedicated readout zone. ZAC utilizes a similar configuration. More details about the configuration can be found in the Appendix \ref{bench-config}. The final mapping solution for each benchmarked quantum circuit can be found in \cite{GCMABench}.

As shown in Figure \ref{fig:results}-(a), we evaluate and compare the total transfer count, the total PTO count, and the total distance traveled by all atoms for various quantum circuits. These three parameters directly influence the execution time and total fidelity of the quantum circuit. The \textit{total} term signifies that the values are respectively summed on all the transitions involved in each scheduled quantum circuit. We compare all three parameters to the ZAC mapper and additionally include a comparison for the PTO counts with the MQT mapper, as it does not provide the other parameters.

The total transfers, the total PTO counts, and the total distances traveled by the atoms in micrometers for each quantum circuit are shown in Figure \ref{fig:results}-(a), (b) and (c), respectively. In Figure \ref{fig:results}-(a), we additionally plot the theoretical limit (Theorem \ref{theorem-glob-opt-subgraph}). Note that this limit is given by $\lambda_q$ and not by $\lambda_q + N_C$, since cycles are very rare cases for the entire set of algorithms tested. Hence, the addition of one or two transfer counts is not visible in the figures. Moreover, $N_C$ depends on the final mapping, thus $N_C$ also depends on the final result given by the GA. Therefore, if the best mapping found in our benchmark contains cycles, alternative mappings can give fewer cycles.

From the results in Figure \ref{fig:results}-(a), we first observe that the proposed general mapping solution always achieves the theoretical minimum of transfers (up to $N_C = 1$ or $N_C = 2$). Although ZAC uses intrinsically bipartite graphs for qubit reuse and transfers, ZAC does not necessarily reach the theoretical minimum of transfers. This result highlights that our graph-theoretic combinatorial optimization fully captures the mapping complexity of quantum circuits, and optimizing the transfer for each transition leads to the global optimum.

Secondly, we observe that there exists a strong link between the total amount of transfers, the total distance traveled by the atoms, and the total PTO count. In fact, it is clear that more transfers implicitly increase both the total distance traveled by the atoms and the total number of PTOs. However, we note that the deviation in the total transfer count between ZAC and the proposed general approach do not explain the deviation in total distance traveled for all the algorithms tested. For example, there is a difference of approximately $100 * (840 / 436 - 1) \simeq 92\%$ in the total transfer count for the QFT (n29) circuit, while there is a difference of approximately $100 * (34434 / 7949 - 1) \simeq 333\%$ in the total distance traveled by the atoms for the same circuit. 

Our advantage comes from the separation of the qubit transfer decision (Section \ref{section_transfers}) and the efficient encoding of sticks/substicks (Section \ref{qubit-mapping-encoding}). This disentanglement allows us to optimize the distance traveled by the atoms through the entire set of stages while keeping the amount of atom transfers constant. In parallel, ZAC does not implement any specific encoding of the problem and considers the stage as successive 2D grids without temporal link. Thus, changing the position of an atom does not influence its positions in the next stages. In addition, ZAC implements a weighted approach of the simulated annealing (SA) in which the distance traveled by the atoms has a greater impact on the early stages of the circuit. Therefore, by setting the inverse of the total distance as the discriminatory value of the GA, our proposed approach finds a better optimization of the total distance over all tested quantum circuits.

Overall, we observe that, except for QFT, all the tested circuits present a few advantages or even disadvantages in terns of PTO count compared to ZAC or MQT mappers. The circuits BV, CAT and GHZ are circuits with simple and similar topologies (i.e. equivalent entanglement complexity), which make them very straightforward to optimize. This simplicity, in terms of optimization, explains the small deviation between the proposed approach and ZAC/MQT.

The usefulness of the GA in our approach can be highlighted when running five quantum circuits for different GA parameters, as shown in Figure \ref{fig:results-special}. For the KNN, Mulitply, WState and SWAP-test circuits, we modify the GA by setting the inverse of the PTO count as the discriminatory value instead of the inverse of the distance for all GA passes \ref{ga-details}. We observe an advantage of the proposed approach compared to ZAC/MQT in terms of the PTO count. However, the KNN, Mulitply and WState circuits present a disadvantage for the total distance. Compared to the results shown in Figure \ref{fig:results}-(b), the SWAP-test circuit optimized with this version of the GA exhibits fewer advantages over ZAC in the total distance traveled by the atoms. The topology of the Ising circuit is such that a dummy placement of the qubits on the grid (before running the GA) is already sufficient to minimize the PTO count. The total transfer count is not modified, or at most by $N_C = 1$ or $N_C = 2$. We clearly obtain fewer PTOs than in \ref{fig:results}-(c). However, the total distance is much less minimized than in Figure \ref{fig:results}-(b) where the GA is run.

The graph framework that underpins our approach leads to fewer PTO counts. We note that despite the MQT mapper providing a better optimization of the PTOs compared to ZAC, we constantly find fewer PTO over the full benchmark when considering the cases where the PTO count is the discriminatory value of the GA. In fact, the efficient encoding of sticks/substicks allows to optimize the PTO count through all the stages in parallel, hence leading to a better approximation of the global optimum.

Comparing Figure \ref{fig:results} to Figure \ref{fig:results-special}, we cannot conclude whether our optimal mappings reduce the global execution time of the quantum circuit since this will depend on the parameters of a given hardware. For example, if SLM/AOD trap switchings are less efficient than the speed of the optical tweezers, then the mapping minimizing the PTO count is more advantageous. By contrast, if the speed of the optical tweezers is less efficient than that of the SLM/AOD trap switching, then the mapping minimizing the traveled distance is more advantageous. More generally, it is recommended to run the proposed mapper multiple times with/without the GA or with different parameters of the GA. The user can then choose the best mapping solution.

The results presented in Figure \ref{fig:results-special}-(a) and Figure \ref{fig:results-special}-(b) highlight the unbalance and trade-off that quantum programmers must be aware of during the compilation process. In fact, one of the advantages of using such global optimum approximators (GA or SA) is the ability to modify the measured discriminatory value according to the needs. For real applications, it is possible to set the distance plus the PTOs as discriminatory values, weighted by factors that represent actual hardware parameters, for example. To go further, given explicit details of the targeted NAQC, it is also possible to set the measured score of the GA/SA as the global fidelity of the circuit or execution time, etc. Although the optimization process can require substantially more computational power in that case, the resulting mapping would be tailored for predictable execution such as e.g. global fidelity and execution time.

\section{Conclusion}

In this work \cite{Gentil2026Patent}, we proposed a mathematical framework and an optimization procedure of the mapping problem for quantum circuits targeting NAQC, where atom transfer is a costly and constrained operation. Our method is general, scalable, circuit topologies aware, and quantum gate-agnostic. Its adaptability makes it promising for a wide range of neutral atom platforms. 

The proposed general mathematical framework is based on graph theory to determine the minimal number of transfers required to execute any quantum circuit. This theory applies to monolithic architectures but also to zoned architectures, an advantageous feature of NAQC platforms where gate operations, readout process, and storage are localized in distinct spatial zones. For future NAQC platforms, it is natural to predict additional zones dedicated to, for example, multi-qubit gates such as the CCZ gate. Those additional zones are naturally incorporated into the developed framework. Then, building on this theoretical foundation, we tackled the practical problem of mapping qubits to physical atom positions to minimize the total traveled distance and the number of PTOs. We introduced a novel encoding of this nonlinear integer problem and applied a genetic algorithm to explore the configuration space.

Our method reveals a trade-off between movement efficiency and transfer parallelism. In fact, we consistently find fewer or the same number of transfers as ZAC \cite{lin2025reuseawarecompilationzonedquantum}, a recent state-of-the-art compiler, across all tested circuits. When optimizing the genetic algorithm for minimizing distance, we achieve significantly shorter traveled distances, albeit without the guarantee of fewer PTOs. In contrast, prioritizing PTO minimization leads to fewer parallel operations compared to both ZAC and MQT mappers, sometimes at the cost of longer transfers.

In practice, tailoring the optimization procedure towards movement efficiency and transfer parallelism can be made quantitatively by weighting the relative cost of PTOs and traveled distance according to hardware-specific parameters, such as the ratio between SLM/AOD switching times and optical tweezer movement speeds. This trade-off provides flexibility to quantum programmers and compiler developers, allowing them to tailor compilation according to the characteristics or constraints of the hardware. More specifically, the optimization objective can be extended to directly approximate execution time or circuit fidelity, enabling the mapper to produce hardware-calibrated solutions at the expense of increased computational effort.

\begin{acknowledgments}
This research is financially supported by the Dutch Research Council (NWO) and the National Growth Fund programme QDNL through the Research programme NGF-Quantum Delta NL 2023, under project number NGF.1623.23.030. We thank the TU/e members of the Neutral Atom Kat-1 Collaboration (www.tue.nl/rydbergQC) for valuable discussions and in particular Raul Parcelas Resina dos Santos and Emre Akaturk for useful comments on the manuscript.
\end{acknowledgments}

\bibliographystyle{quantum}

\onecolumn
\appendix

\section{Graph Theory Of Optimum Qubit Transfers}
\label{gt-of-opt-transf}

The number of instructions involved inside a stage $s_i$ is $N_{\theta,i}$. An instruction $\theta_{i, j}$ has a set of qubits $Q_{i, j} \defeq \{q_k, q_l, q_m,\, ...\}$. An instruction containing only a single qubit gate ensures that $|Q_{i, j}| = 1$, whereas an instruction containing multi-qubit gates has the lower bound $|Q_{i, j}| \geqslant 2$. The set of qubits exchanged from $\theta_{i, j}$ to $\theta_{i + 1, k}$ is therefore $\Gamma_{i,j,k} \defeq Q_{i, j} \; \cap \; Q_{i+1, k}$. Because each qubit is associated with exactly one instruction, there exists a set $\Theta_{i} \defeq \{\theta_{0, j_0}, \;...\; , \theta_{N_s - 1, j_{N_s - 1}}\}$ for each qubit where $j_k \in \llbracket 0; N_{\theta,k} - 1\rrbracket$. 

The general quantum digraph $G(V, E)$ representing the circuit can be summarized as follows:
\begin{itemize}
    \item Each instruction $\theta_{i,j}$ is a vertex: $V \defeq \{\theta_{i,j}\}$
    \item An edge from $\theta_{i, j}$ to $\theta_{i+1, k}$ exists iff $\Gamma_{i,j,k} \neq \varnothing$
    \item Each existing edge $e = (\theta_{i, j}, \theta_{i+1, k})$ has a weight $\omega(e) \defeq |\Gamma_{i,j,k}|$
\end{itemize}
where $i \in \llbracket0; N_s-1\rrbracket$ and $ j,k \in \llbracket 0; N_{\theta,i} - 1\rrbracket$. From the sequences of instructions $\Theta_i$, there exists a subgraph $P_i(V_{P_i}, E_{P_i})$ for each qubit $q_i$:
\begin{align}
    V_{P_i} = \;&\Theta_{i} \\
    E_{P_i} = \;&\{(\theta_{0, j_0}, \theta_{1,j_1}), \; ... \; , (\theta_{N_s - 2, j_{N_s - 2}}, \theta_{N_s - 1, j_{N_s - 1}})\})
\end{align}
where $j_k \in \llbracket 0; N_{\theta,k} - 1\rrbracket$. Note that $|P_i| \defeq |V_{P_i}| = N_s$. By definition, the subgraphs $P_i$ are paths of $G$ and represent the successive instructions that each qubit undergoes through the stages.

\begin{lemma}
\label{lemma-matching}
    A matching $M_{i}$ of $H_i$ is a valid physical transfer configuration from stage $s_i$ to the next stage $s_{i+1}$.
\end{lemma}
\begin{proof}
    We recall that, in a graph $G'(V',E')$, a matching is a subset $M \subseteq E'$ such that no two edges in $M$ share a common vertex. Moreover, we define a coloring for a given edge $(\theta_{i, j}, \theta_{i + 1, k})$:
    \begin{itemize}
        \item The edge is colored black if and only if (iff) all the qubits in $\Gamma_{i,j,k}$ are static.
        \item The edge is colored red iff all the qubits in $\Gamma_{i,j,k}$ are dynamic.
    \end{itemize}
    Note that an edge with, simultaneously, static and dynamic qubits is physically a nonsense.
    
    In $H_i$, there are three different cases:
    \begin{enumerate}
        \item The vertices $\theta_{i, j}$ and $\theta_{i + 1, k}$ of a given edge are not shared with any other edges. Then, the edge is equivalently colored black or red.
        \item Two or more edges share the same endpoint $\theta_{i + 1, k}$. Let us assume that two or more of those edges with endpoint $\theta_{i + 1, k}$ are colored black. Since all the instructions in $s_i$ have a unique physical position, it means that at least two qubits of different instructions $\theta_{i, j}$ and $\theta_{i, j'}$ must reach the same target instruction $\theta_{i+1, k}$, while being static. It is absurd, and, by contradiction, at most one edge must be colored black.
        \item Two or more edges share the same vertex $\theta_{i, j}$. Similarly, the same conclusion as for the second case holds: two qubits cannot move away from each other while being static and at most one edge is colored black.
    \end{enumerate}
    We deduce that a configuration is valid iff each vertex in $H_i$ has at most one incident black edge. Therefore, the set of edges colored black is, by definition, a matching.
\end{proof}

The score of a matching $M_{i}$ is defined as the amount of static qubit assigned by the edges in $M_{i}$:
\begin{equation}
    \Omega(M_i) = \sum_{e \in {M_i}} \omega(e)
\end{equation}
Depending on the topology of $H_i$, several matchings $M_{i}$ can exist. We relabel them $M_{i,0}$, $M_{i,1}$, etc.
\begin{lemma}
\label{lemma-local-opt}
    The minimal amount of qubit transfers, that is, the maximum amount of static qubits, from $s_i$ to $s_{i+1}$ is given by the maximum weighted matching $\mathcal{M}_i$ of $H_i$. $\mathcal{M}_i$ is not necessarily unique and we note them $\mathcal{M}_{i,0}$, $\mathcal{M}_{i,1}$, etc.
\end{lemma}
\begin{proof}
    The amount of transfers for each $H_i$ with matching $M_{i,j}$ is given by $N_q - \Omega(M_{i,j})$: this is the actual qubit transfer cost for transition $s_i \rightarrow s_{i+1}$. Hence, maximizing $\Omega$ allows one to find the minimal amount of qubit transfers. The maximum of $\Omega$ is given by $\text{max}_{M \in \{M_{i,0}, M_{i,1}, ...\}}(\Omega(M)) = \mathcal{M}_i$ which is by definition the maximum weighted matching of $H_{i}$. In the case of multiple equivalent maximum weighted matching $\mathcal{M}_{i,0}$, $\mathcal{M}_{i,1}$, etc, then $\Omega(\mathcal{M}_{i,0}) = \Omega(\mathcal{M}_{i,1}) = \,...$ by definition.
\end{proof}

Lemma \ref{lemma-local-opt} gives the local optimum for the transfer of qubits. It is then possible to find the global optimum, i.e. the minimal amount of transfers for the whole quantum circuit. For a given $H_i$, we arbitrarily select $\mathcal{M}_i$ among $\mathcal{M}_{i,0}$, $\mathcal{M}_{i,1}$, etc., since each of those maximum weighted matchings has the same score.

We give the proof for Theorem \ref{theorem-global-opt-m}:
\begin{proof}
    Given lemma \ref{lemma-local-opt}, it suffices to prove that the $\mathcal{M}_i$ are independent of each other. In fact, the choice of the static/dynamic qubits from $s_i$ to $s_{i+1}$ does not influence the sets $Q_{i+1,j}$ or $Q_{i+2,j}$. Moreover, the sets $Q_{i,j}$ fully determine the graphs $H_i$ by definition. Hence, the graphs $H_i$ are independent of each other. As every matching $\mathcal{M}_{i}$ is only determined by the topology of $H_i$, the matching $\mathcal{M}_{i}$ is also independent of each other matchings $\mathcal{M}_{j}$. By independence of the local optimums, the global optimum is obtain by summing the local optimum.
\end{proof}

After generalization of Theorem \ref{theorem-global-opt-m} for Z-arch, the proof of Theorem \ref{theorem-glob-opt-subgraph} follows:
\begin{proof}
We consider the underlying undirected graph of $G$. For every instruction $\theta_{i, j}$ in $Z_m$ and $\theta_{i', j'}$ in $Z_n$ with $m \neq n$, the set that contains the existing edges $e = \{\theta_{i, j}, \theta_{i', j'}\}$ is defined as $E_{m,n}$. The edges $E_{m,n}$ connect the subgraphs $G_m$ and $G_n$ such that $E_0 \cup E_{0,1} \cup E_{0,2} \cup \;... \cup E_1 \cup E_{1,2}\; ... = E$. Note that the sets $E_{m,n}$ and $E_{n,m}$ are equivalent.

Let us assume that the number of transfers is globally minimized in $G$ and the edges are already colored black/red. We put the edges colored in red (transfer) in the set $T$ and the edges colored in black (static) in the set $S$. There are two cases:
\begin{enumerate}
    \item $\forall (p,q) \in \llbracket 0; N_Z - 1\rrbracket^2, p \neq q,\, E_{p,q} \subseteq T$. In other words, all edges connecting the subgraphs $G_n$ have their underlying qubits defined as dynamic. The transfer cost associated with the sets $E_{m,n}$ is then $\sum_n\sum_{m \neq n}\sum_{e \in E_{m,n}} \omega(e)$. This is the sum of the weights of each of the edges that connect the subgraphs $G_n$. The optimal configuration transfer found is therefore equivalent to applying theorem \ref{theorem-global-opt-m} to each subgraph $G_n$. The minimum amount of dynamic qubit for each $G_n$ is defined as $\lambda_{q,n}$. The global minimum of dynamic qubits is then:
    \begin{equation}
        \lambda_q = \sum_{n} \lambda_{q,n} + \sum_n\sum_{m \neq n}\;\;\sum_{e \in E_{m,n}} \omega(e)
    \end{equation}
    \item $\exists (p,q) \in \llbracket 0; N_Z - 1\rrbracket^2, p \neq q,\, E_{p,q} \cap S \neq \varnothing$. In other words, there is at least one edge connecting the zones that is defined as static: for example, the qubits in the set $\Gamma_{i,j,k}$ undergo an instruction $\theta_{i,j}$ in a zone $Z_m$ and a instruction $\theta_{i+1,k}$ in a zone $Z_n$ for the next stage without any movement. It is physically impossible since the zones are spacially separated by definition of Z-arch. Hence, this transfer configuration found globally in $G$ is invalid.
\end{enumerate}
In summary, in order to find a valid transfer configuration globally in $G$, Theorem \ref{theorem-global-opt-m} must be applied independently to each subgraph $G_n$.
\end{proof}

\section{Stick Combinations}
\label{substick-combination}

\begin{figure}
    \centering
    \includegraphics[width=0.4\linewidth]{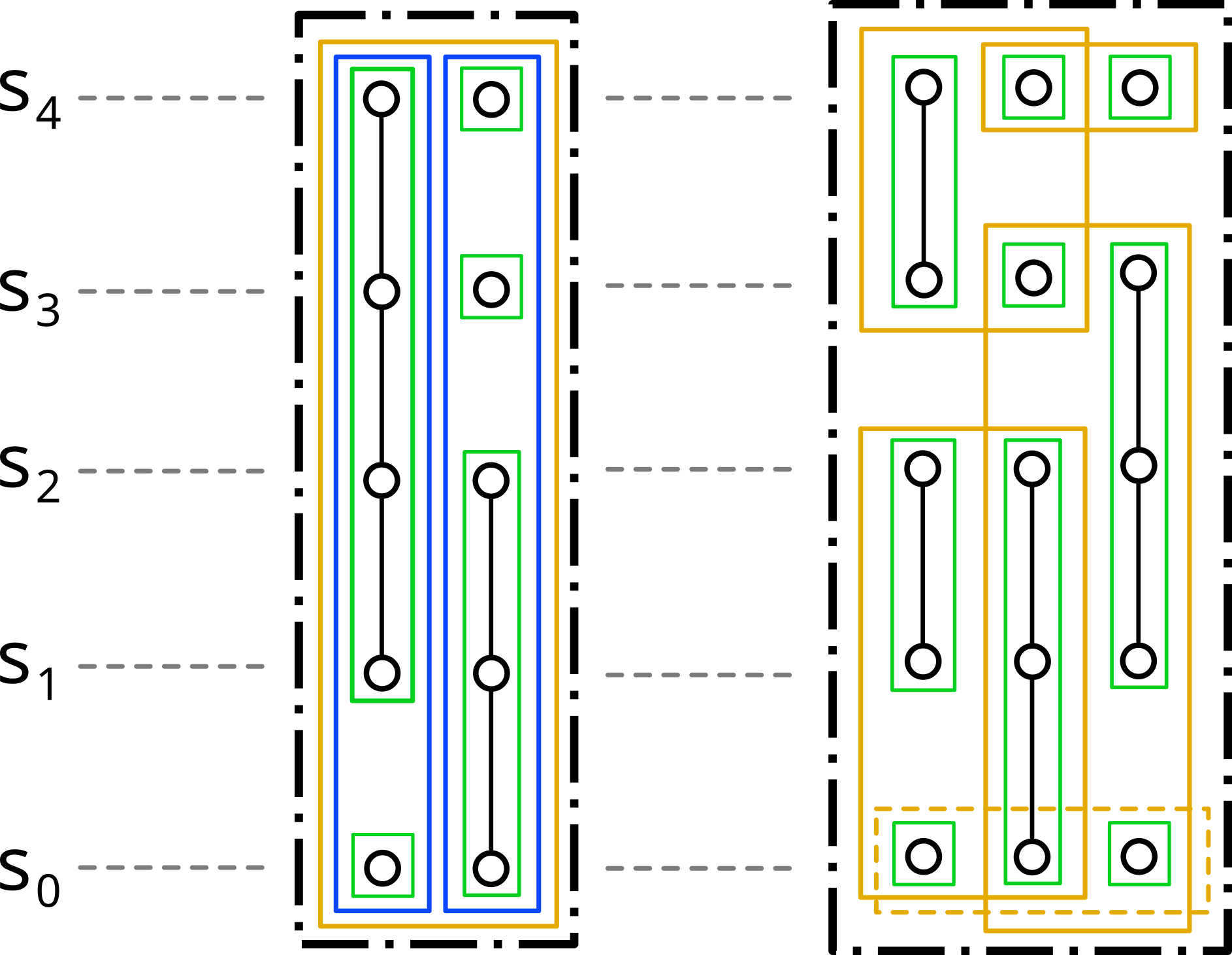}
    \caption{The dashed black box on the left is an example of stick where the substicks are represented as green boxes. The orange boxes are the \textit{forks}. The blue boxes are the \textit{tines}. Only the tines constituting the forks can be permuted. Thus, the combinations are encoded into the GA by a Boolean variable. The dashed black box on the right is an example of stick for three-qubit instructions / gates. Only the forks are represented for simplicity. The total amount of possibilities for the combinations increases non-linearly with the amount of qubits per stage. Note that the fork in the bottom of the three-qubit stick is dotted: this fork includes the two qubits on the right/left, not the middle qubit.}
    \label{fig:stick-comb}
\end{figure}

An example of a stick, represented by dashed black lines, is shown on the left of Figure \ref{fig:stick-comb}. This stick is associated with two-qubit instructions. Therefore, this stick lies in the entanglement zone. The length of this stick is 4. It contains 5 substicks highlighted with green boxes. The orange boxes are called \textit{forks} and the blue boxes are called \textit{tines}. We can see that for each transition $s_i \rightarrow s_{i+1}$ at least one atom is static. Therefore, the only way to shuffle forks is to permute the tines. Hence, the combination in this case is limited to a binary choice represented by a Boolean variable. Initially, the tines are randomly placed right or left, and the Boolean variable determines whether the initial placement is permuted or not. Hence, a stick in the storage zone can contain only one substick with one fork and one tine: no combination is available. The same would apply to a hypothetical readout zone.

The idea of combinations can be generalized to sticks associated with multi-qubits instructions. An example is shown on the right of Figure \ref{fig:stick-comb}. The number of forks can grow drastically depending on the topology of the stick. The general definition is the following.
\begin{definition}
\label{def-fork}
Any set of substicks that lies on exactly two tines and \textit{fully fills} the given tines is a fork.
\end{definition}
\textit{Fully filled} means that given a stage $s_i$, exactly two different substicks are involved in $s_i$: there is no absence of substick in any of the tines of any of the given stage. Therefore, the complete description of a stick, apart from its position, is a list of Boolean variables, i.e. one variable per fork. 

The definition \ref{def-fork} gives an idea for a general algorithm to find all the forks of a given stick. However, executing and evaluating quantum circuits involving multi-qubit gates (with more than two qubits) is beyond the scope of the present work. Hence, no algorithm or benchmark is provided for general multi-qubit gates.

\section{Genetic Algorithm}
\label{ga-details}

The total traveled distance and PTO estimation are non-linear constraints. Together with the chosen encoding, the mapping problem forms a nonlinear integer problem (NLIP). NLIPs are typically NP-hard, tending to confirm that the mapping problem is also NP-hard. General nonlinear integer programming problems are NP-hard, and exact methods typically become computationally prohibitive for large instances. Therefore, metaheuristics are commonly used in practice to compute satisfactory solutions\cite{Tian1998}. The GA approach (without crossover) can be summarized as follows.

A \textit{generation} is by definition the loop "Step 2 to Step 5" processed once. Compared to SA, there is no user-defined \textit{heat} value: the GA iteratively tries new solutions, but saves those new solutions only if their score is higher compared to \textit{the worst among the trend}. To stop the main loop, it can then be advantageous to wait for the global performances of the solutions to converge instead of arbitrarily aborting the loop. Technically, it means that while keeping track of the average score of all the current solutions, the user can wait for the program to stabilize the average score. 

GAs have the inconvenience of scaling the memory consumption very rapidly with the complexity of the problem. However, since the GA method analyzes a relatively broad part of the search space, it is well known to rapidly give a \textit{good}/\textit{acceptable} solution compared to SA.

The parameters that can be configured in the GA are:
\begin{itemize}
    \item The number of members $N$ in the initial population.
    \item The number of children $M$ for each solution.
    \item The score can be based purely on the inverse of the global distance, purely on the inverse of PTO count, or on both, potentially weighted to prioritize the optimization on the total distance or the PTO count.
    \item For randomization, only the sticks can be moved, or the sticks can only be shuffled internally, or both.
    \item The maximum distance to which the sticks can be moved. Note that the position is bounded by the size of the grid.
    \item The number of sticks that are either randomly moved or internally shuffled for each child.
\end{itemize}
Those parameterizations influence the convergence speed and the search space of the GA. They \textit{must} be tuned in order to get the best performances depending of the quantum circuit depth and width.

\begin{figure}
\caption{GA Optimization Procedure}
\begin{algorithmic}[1]
\Require Population size $N$, number of children $M$, max iterations $T$
\Ensure Best solution found

\State Generate an initial population of $N$ solutions

\For{$t = 1$ to $T$}
    \State Generate $M$ children for each solution in the population
    \Statex \quad (Each child is a randomized variant of its parent by moving or shuffling sticks)

    \State Evaluate all solutions
    \Statex \quad (Score = inverse of total distance or total PTO count)

    \State Select the best $N$ solutions
    \Statex \quad (Discard the remaining $N \times M$ worst solutions)
\EndFor

\State \Return the solution with the highest score
\end{algorithmic}
\end{figure}

\section{Graph Theory For Transfer Operations}
\label{appendix-gt-for-to}

From a stage $s_i$ to a stage $s_{i + 1}$, we consider that we found the optimal set of atoms to be transferred. Assuming that a certain mapping is applied, each atom $q_l$ starts at a position $\overrightarrow{a_l}=(a_{l,x}, a_{l,y})$ at $s_i$ and ends at a position $\overrightarrow{b_l}=(b_{l,x}, b_{l,y})$ at $s_{i + 1}$. We define $A_i = \{\overrightarrow{a}_0, \overrightarrow{a}_1, ...\}$ and $B_i = \{\overrightarrow{b}_0, \overrightarrow{b}_1, ...\}$. To avoid physical collisions, the starting positions must respect $\overrightarrow{a}_0 \neq \,...\, \neq \overrightarrow{a}_{N_q}$. The same applies to the elements of $B_i$. The transfer digraph $G_{T, i}(V_{T, i}, E_{T, i})$ for the transition $s_i \rightarrow s_{i + 1}$ is defined as $V_T \defeq A_i \cup B_i$ and $E_T \defeq \{(\overrightarrow{a}_0, \overrightarrow{b}_0), (\overrightarrow{a}_1, \overrightarrow{b}_1), ...\}$. It is important to note that from this definition an atom is transferred at most once, i.e. an atom $q_l$ is associated with at most one couple $(\overrightarrow{a}_l, \overrightarrow{b}_l)$, before post-processing of the potential cycles (see Section \ref{section-transfer-operation}). From the definition of $G_T$, we deduce the following lemma:
\begin{lemma}
\label{lemma-inner-outer-gt}
    Every vertex of $G_T$ has at most one inner edge and one outer edge.
\end{lemma}
\begin{proof}
    Consider that there exists a vertex $v$ with multiple inner edges. The vertex $v$ belongs to $B_i$, that is, $v$ is an ending position. Physically, it means that multiple atoms are transferred to the position described by $v$. Since every atom is transferred at most once, a physical collision happens in this case. Hence, there is no vertex in $G_T$ with multiple inner edges. An equivalent argument holds for multiple outer edges.
\end{proof}

We can immediately state the following Lemma.
\begin{lemma}
\label{lemma-gt-non-int-paths}
    A transition digraph $G_T$ is the union of non-intersecting paths.
\end{lemma}
\begin{proof}
    A path is a set of successive edges. Two paths with one or more vertices in common are said to intersect. Note that a cycle is a path, and two edges have a single vertex in common, but it does not intersect with itself. For two or more paths to intersect at a vertex $v$, the vertex $v$ must have multiple inner or outer edges. By Lemma \ref{lemma-inner-outer-gt}, this can never happen in $G_T$. Hence, only non-intersecting paths exist in $G_T$.
\end{proof}

Therefore, a simple algorithm can build the sequence of transfer operations required for the transition $s_i \rightarrow s_{i + 1}$ to be executed on actual NAQC hardware. It starts from the leaves (vertices without outer edge) of each path of $G_T$ and back-propagates to the roots (vertices without inner edge). The set of edges collected along each path is a time-ordered sequence of transfer operations which avoids physical collisions of the atoms.

Given a transfer digraph $G_T$, each path does not necessarily have a root and a leave. In that case, those paths without leaf/root are \textit{cycles}. Note that except for cycles, paths necessarily have a couple of leaf/root. The main issue is that any order of execution of the transfers in a cycle leads to a physical collision. Hence, it is necessary to \textit{break} all the cycles contained in $G_T$ before creating the sets of transfer operations.
 
The operation to break a cycle is the following. It suffices to transfer a single atom to an unused position on the grid ahead of any other transfers. For example, assume that a cycle $C_T$ is included into $G_T$. We arbitrarily select any vertex $v$ in $C_T$. By the definition of a cycle, $v$ has both a unique inner edge and a unique outer edge. Thus, $v$ is associated with the couple $(\overrightarrow{a}_l, \overrightarrow{b}_l)$ and the unique predecessor where, in graph theory, the predecessors of $v$ are the vertices $v'$ that have their outer edge pointing to $v$. Here, $v'$ is unique and is associated with the couple $(\overrightarrow{a}_{l'}, \overrightarrow{b}_{l'})$. Note that $\overrightarrow{b}_{l'} = \overrightarrow{a}_{l}$ by construction. Now, assume that the position $\overrightarrow{x}$ is not occupied by any atom at stage $s_i$ and stage $s_{i + 1}$, that is, $\overrightarrow{x} \notin A_i \cup B_i$. To break the cycle $C_T$, we apply the transformation $(\overrightarrow{a}_l, \overrightarrow{b}_l) \rightarrow (\overrightarrow{x}, \overrightarrow{b}_l)$: the vertex $v$ is \textit{moved} to position $\overrightarrow{x}$. To be physically consistent, the atom that originally started at $\overrightarrow{a}_l$ must be transferred to the position $\overrightarrow{x}$ before executing the transfer $(\overrightarrow{x}, \overrightarrow{b}_l)$. We observe that, discarding the edge $(\overrightarrow{a}_l, \overrightarrow{x})$, the subgraph $C_T$ is no longer a cycle since $v'$ became a leaf and $v$ a root.

More generally, the types of transfer $(\overrightarrow{a}_l, \overrightarrow{x})$ required to break cycles are called \textit{additional transfers} and must be executed on top of any other transfer operation to avoid any physical collision. Once all cycles have been broken, by Lemma \ref{lemma-gt-non-int-paths}, it is possible to estimate a physically feasible sequence of transfer operations.  We give the following proof for Theorem \ref{theorem-break-cycles}:
\begin{proof}
Since only one additional transfer is required to break one cycle, $N_C$ additional transfers are necessary to break $N_C$ cycles. As the additional transfers are performed prior to all the other transfers, the total amount of transfers is $\lambda_q + N_C$.
\end{proof}

\section{Atom Routing}
\label{appendix-grid-routing}

In the present work, we often refer to the Euclidean distance for the transfer of atoms. This implies that the atoms travel \textit{ straight} from the starting position to the ending position. Consider the atomic grid completely full of atoms except one slot. An atom must move from its current position to the empty slot. It is not possible with real hardware to move this atom in a straight line without a possible collision with other atoms. Therefore, we implicitly assume that enough space is provided between the slots in order for the transferred atom to travel without collision. This space is called \textit{travel space}. The constraints of the transfer digraph $G_T$ defined in Appendix \ref{appendix-gt-for-to} show that several atoms transferred in parallel cannot collide. In addition, if an atom is traveling in an empty area of the grid, the atom can move in a straight line. The exact routing, either in \textit{Manhattan} style or straight line, can be post-processed after an optimal mapping has been found in the output of the GA. An example of routing is illustrated in Figure \ref{fig:routing-example}.

\begin{figure}
    \centering
    \includegraphics[width=0.3\linewidth]{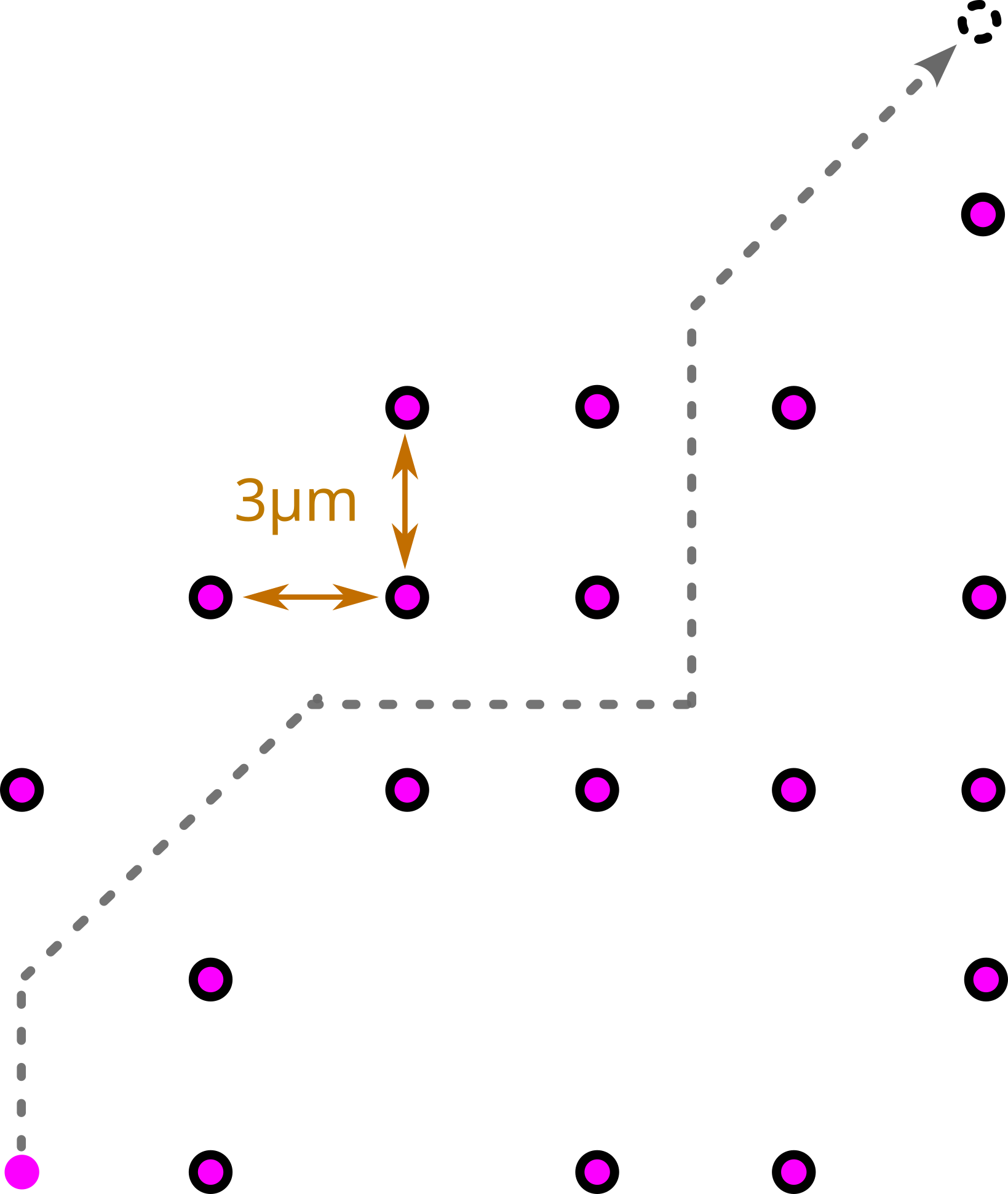}
    \caption{Example of grid routing that can be executed on real hardware. The atoms are purple dots. The black circles are static atoms. The dotted black circle is the ending position of the transferred atom. The dotted gray line is the routing. The 3µm distance separating the atoms in the storage zone is sufficient, in this example, for the transferred atom to travel in between. Straight movements, instead of the \textit{Manhattan style}, can be post-processed once the final mapping has been generated.}
    \label{fig:routing-example}
\end{figure}

Despite these considerations, the distance traveled by the atoms in the GA is based on the Euclidean distance only, since this measurement approach leads to circuit mappings with shorter total distances. The ZAC's benchmark is also based on the Euclidean distance. Note that in real hardware, the \textit{travel} space between slots depends on the precision of the optical tweezers.

\section{PTOs Sequence}
\label{pto_sequence}

\begin{figure}
\caption{Construction of Parallel Transfer Operations (PTOs)}
\label{fig:algoptos}
\begin{algorithmic}[1]
\Require Ordered set of transfer edges
$\mathcal{E} = \{e_1,e_2,\ldots,e_n\}$
\Ensure Ordered set of physically valid PTOs
$\mathcal{P} = \{PTO_1,PTO_2,\ldots\}$

\State Initialize $\mathcal{P} \gets \emptyset$

\ForAll{$e \in \mathcal{E}$}
    \State assigned $\gets$ false

    \ForAll{$PTO_i \in \mathcal{P}$}
        \If{$e$ is compatible with all transfers in $PTO_i$}
            \State Add $e$ to $PTO_i$
            \State assigned $\gets$ true
            \State \textbf{break}
        \EndIf
    \EndFor

    \If{assigned is false}
        \State Create new PTO $PTO_{\mathrm{new}}$
        \State Add $e$ to $PTO_{\mathrm{new}}$
        \State Append $PTO_{\mathrm{new}}$ to $\mathcal{P}$
    \EndIf
\EndFor

\State \Return $\mathcal{P}$
\end{algorithmic}
\end{figure}

In order to fully utilize the hardware capabilities of NAQC platforms and reduce the execution time of quantum circuits, it is preferable to group compatible transfers into PTOs. At the end of the day, the operations required for quantum engineers to execute a remapping between two successive stages are encompassed into a single time-ordered sequence of PTOs. A PTO is technically a set/list of transfers which might be executed in parallel. Note that a PTO can contain only one transfer. Naming a single transfer operation a PTO is overqualified, but is still a valid "parallel" transfer operation.

As stated in Section \ref{section-transfer-operation}, PTOs are compatible under certain conditions given by the constraints of the AODs that affect the optical tweezers. Here, we show a specific execution of an algorithm that determines the sequences of PTOs for a given transition. We do not claim this algorithm to be optimal in the amount of PTOs but it has the property to not increase the amount of transfers evaluated prior to the PTO estimation, just after cycles removal. In fact, by creating more additional transfers than required, it might be possible to further reduce the length of the PTO sequence. It is then a trade-off to estimate whenever the transfer fidelity or a single PTO is more costly from a hardware point of view.

It is assumed that all cycles have been removed from $G_T$. Then, each path in $G_T$ exhibits a time-ordered sequence of transfer. In terms of programming, this is a \textit{list of transfer lists} , where each path has its own executable list of transfers. The exact algorithm is described in Figure \ref{fig:algoptos}.

As an example, we refer to the example given in Figure \ref{fig:pto-algo-ex}. The actions performed by the algorithm to determine the PTOs starts as follows. Firstly, the edges $e_{0, 0}$ and $e_{1, 0}$ are tested. Since they are compatible, they are pushed into the same PTO, named PTO-0. Then, the edge $e_{0, 1}$ is tested with the two transfers in the current PTO-0. This edge $e_{0, 1}$ is incompatible; therefore, a new PTO, PTO-1 is created that includes the transfer associated with $e_{0, 1}$. The algorithm continues with edge $e_{1, 1}$ and the sequence PTO-0, PTO-1, PTO-2, PTO-3, PTO-4 is obtained. This sequence is a physically valid sequence of parallel transfer operations. Note that $e_{0, 3}$ or $e_{1, 3}$ can be processed first to obtain a similar result but a different time-ordered sequence. In the case where $e_{1, 3}$ is first processed, we obtain PTO-0, PTO-1, PTO-3, PTO-2, PTO-4.

\begin{figure}
    \centering
    \includegraphics[width=0.45\linewidth]{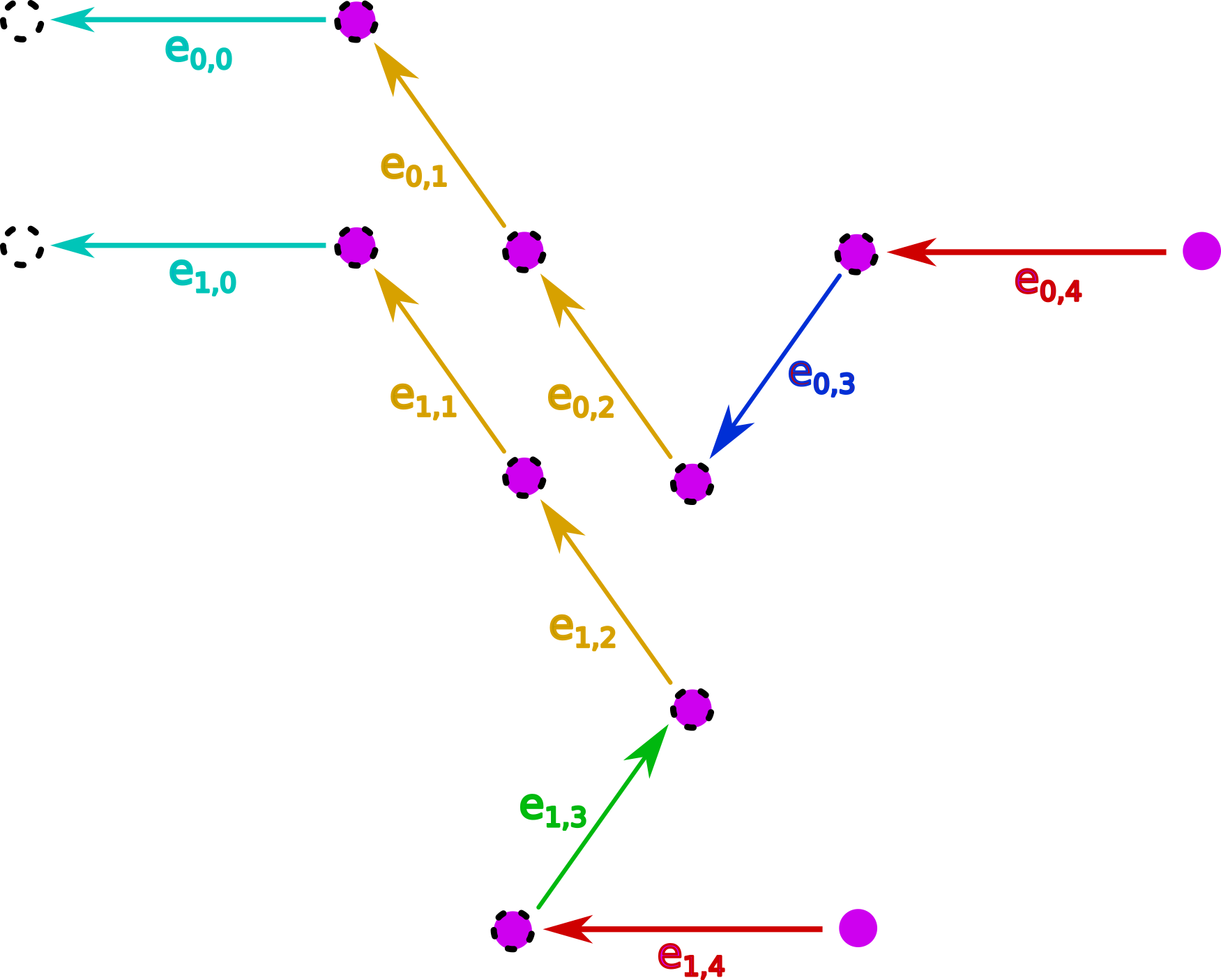}
    \caption{Example of parallel transfer operations (PTOs). The atoms are purple dots. The ending positions are dashed black circles. Each color is associated to a PTO. The transfers contained in a PTO can be executed in parallel according to the AOD constraints. The time-ordered sequence is: \textit{cyan edges}, \textit{gold edges}, equivalently \textit{green edge}/\textit{blue edge} or \textit{blue edge}/\textit{green edge}, then \textit{red edges}.}
    \label{fig:pto-algo-ex}
\end{figure}

The amount of PTOs is dependent on the starting and ending positions of the atoms in the transition $s_i \rightarrow s_{i+1}$. Thus, reducing the total amount of PTOs to execute quantum circuits requires one to carefully map the positions of the atoms at each stage. In addition, we observe that $\lambda_q$ is a theoretical lower bound. Thus, reaching this lower bound is not guaranteed but depends on the cycles created during the mapping of the stages. Even if cycles are rare cases, at least in our benchmark, it might be advantageous to optimize the number of cycles for each stage.

\section{Benchmark Configuration}
\label{bench-config}

The atomic grid is configured for a 2Z-arch, with entanglement and storage zones without a dedicated readout zone. The atomic grid is configured with the same hardware parameters as those evaluated for ZAC. Note that both single-qubit gates and CZ gates can be executed in the entanglement zone whereas only single qubit gates can be executed in the storage zone. In this benchmark, we consider the storage zone as the readout zone. Hence, two stages are added to all quantum circuits. The first additional stage is placed at the beginning of the scheduled quantum circuit: All atoms start in the storage zone. The second additional stage is placed at the end of the scheduled quantum circuit. All remaining atoms in the entanglement zone are moved back to the storage zone for readout. This approach is also used in ZAC \cite{lin2025reuseawarecompilationzonedquantum}. The grid is explicitly described in Figure \ref{fig:grid-config}. Note that compared to ZAC, we consider fewer rows for the storage and entanglement zones: since we consider algorithms with less than 100 qubits, the number of slots available for both storage and entanglement zones is sufficient.

\begin{figure}
    \centering
    \includegraphics[width=0.5\linewidth]{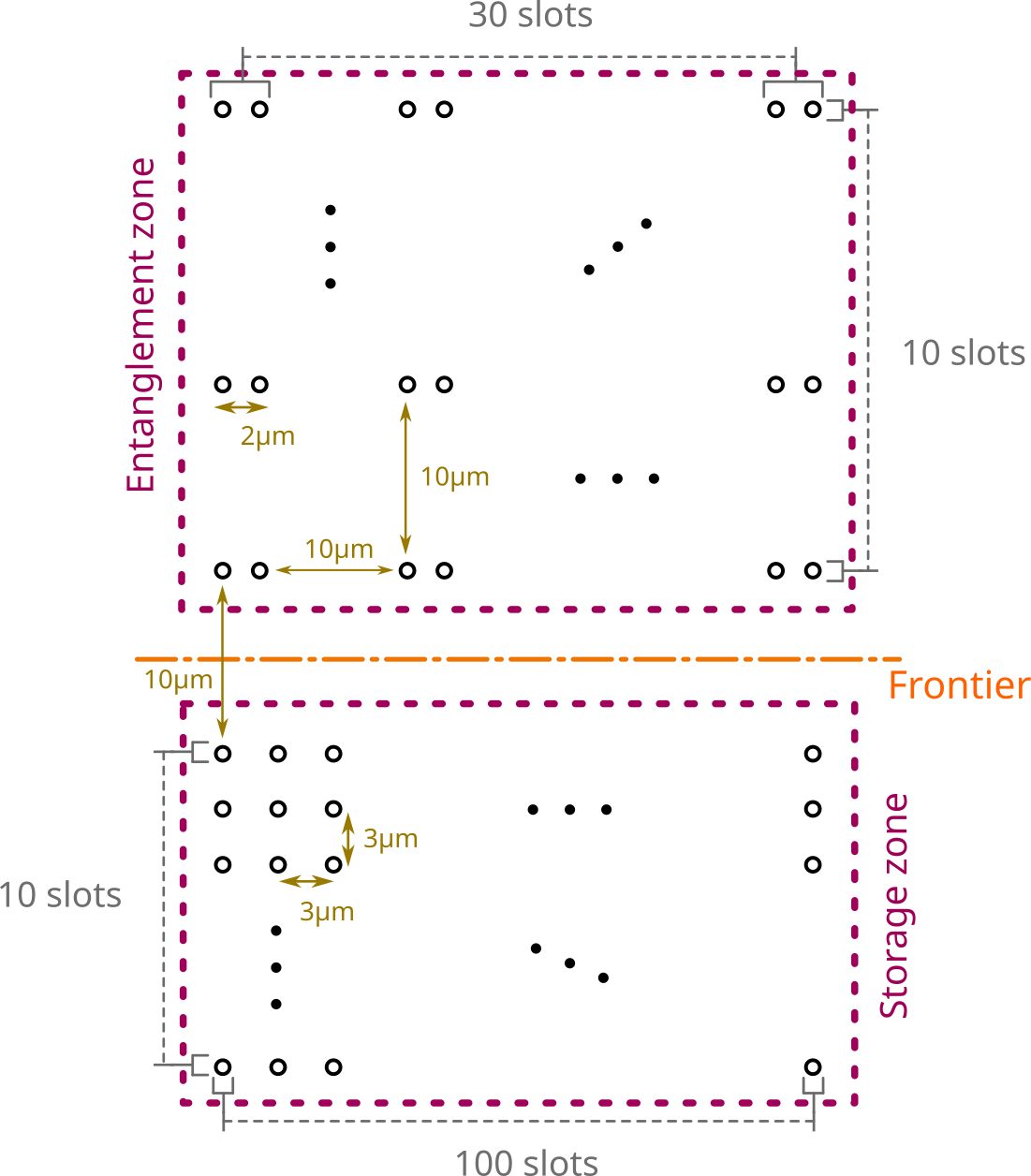}
    \caption{Atomic grid configuration using for the benchmark. Spacing between slots in the storage zone is 3µm: 100x10 slots are available. Spacing between couple of slots in the entanglement zone is 8µm: 30x10 couple of slots are available. The space separating the storage and entanglement zone is 10µm.}
    \label{fig:grid-config}
\end{figure}

The studied quantum circuits are downloaded from QASMBench \cite{li2022qasmbenchlowlevelqasmbenchmark} and are all transpiled with Qiskit using optimization level 3. The output files consist of single-qubit gates U or P and CZ gates. In Figure \ref{fig:schedule-ex}, we show how a quantum circuit is scheduled for a 2Z-arch after transpilation. Note that this hardware configuration involving the global Rydberg laser to address the full entanglement area forces all the CZ gates of a given stage to run in parallel. The stage $\{ \text{CZ}(q_0, q_1), \text{CZ}(q_1, q_0) \text{CZ}(q_2, q_3) \}$ is not valid  since $q_2$ and $q_3$ would undergo another Rydberg excitation while $\text{CZ}(q_1, q_0)$ is processed. However, $\{ \text{CZ}(q_0, q_1), \text{CZ}(q_1, q_0), \text{CZ}(q_2, q_3), \text{CZ}(q_3, q_2) \}$ is valid. Note that the scheduler used for our benchmark is only valid for a 2Z-arch. In our case, we use the \textit{as soon as possible} (ASAP) approach, where the gates are simply processed in the order of the provided sequence. It is important to note that this scheduling process gives similar stages (in terms of entanglement gates) as processed by ZAC in its best usage case (SA + dynPlace + reuse).

\begin{figure}
    \centering
    \includegraphics[width=0.5\linewidth]{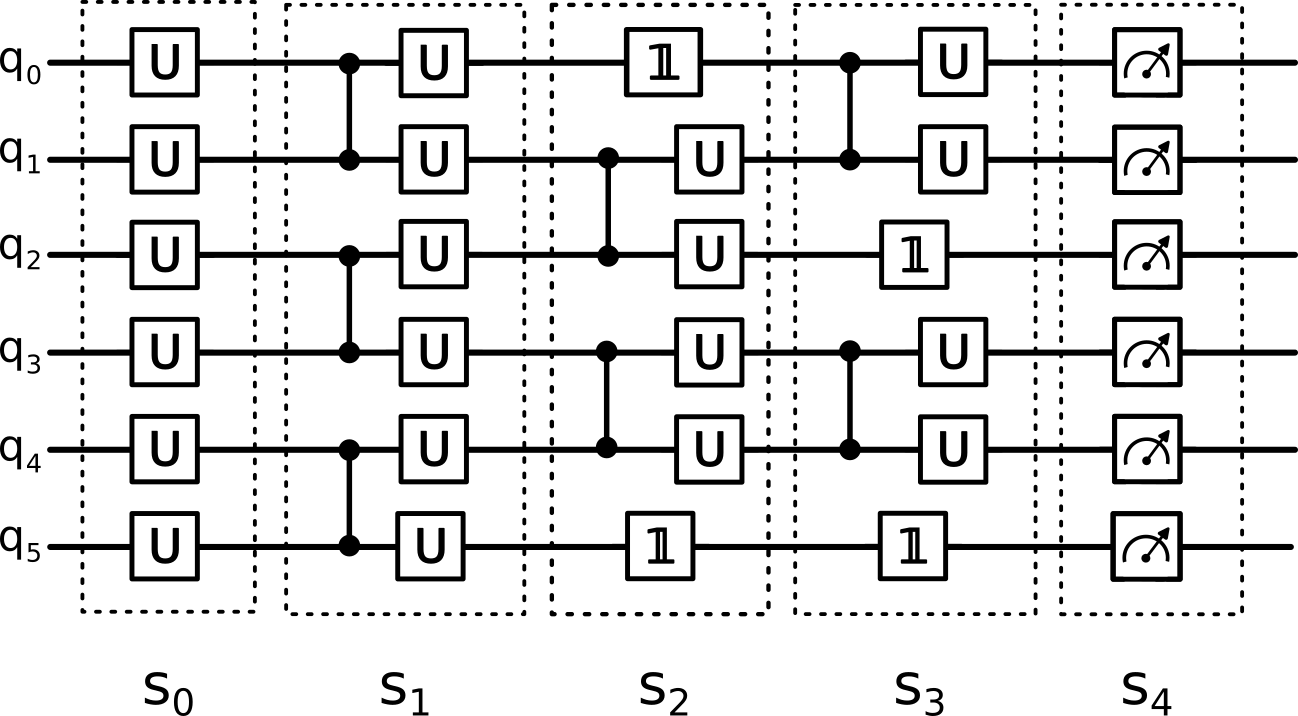}
    \caption{Quantum circuit from Figure \ref{fig:qcirctom} where the gates are expressed in the gate library available for the current benchmarked platform. Note that the stages are similar with Figure \ref{fig:qcirctom}.}
    \label{fig:schedule-ex}
\end{figure}

The GA is split in three distinct passes:
\begin{enumerate}
    \item Move one single stick with a maximum distance of 5 slots in the atomic grid. The inverse of the global distance is measured as the score. 200 generations are processed.
    \item Move one single stick with a maximum distance of one slot in the grid. Shuffle the substicks of a single random stick in the entanglement zone (note that sticks in the storage zone own only one substick, hence shuffling the combinations is irrelevant). The inverse of the global distance is measured as the score. 100 generations are processed.
    \item Shuffle the substicks of a random stick in the entanglement zone. The inverse of the PTO count is measured as the score. One hundred generations are processed.
\end{enumerate}

Note that the best solutions of the first pass are used as initial population for the second pass. The first and second passes find potentially good candidates on the basis of the distance, whereas the third pass refines the solutions on the basis of the PTOs. For the GA, we set the following empirical parameter choices, which seem to give the best convergence speed for quantum circuits with up to 100 qubits: $N = 20$, $M = 2$. The total transfer count, total PTO count, and total distance traveled are averaged over four runs of the GA for each tested quantum circuit. Note that the GA rarely produces cycles for transitions, which implies that the total transfer count rarely changes over the multiple runs of the GA.

\end{document}